\documentclass[twoside]{article}

\usepackage{arxiv2024}
\usepackage[round]{natbib}

\usepackage[utf8]{inputenc} 
\usepackage[T1]{fontenc}    
\usepackage{hyperref}       
\usepackage{url}            
\usepackage{booktabs}       
\usepackage{amsfonts}       
\usepackage{nicefrac}       
\usepackage{microtype}      
\usepackage{xcolor}         

\usepackage{times}
\usepackage{bm}
\usepackage{amsmath}
\usepackage{amsfonts}
\usepackage{amsthm} 
\usepackage{amssymb}
\usepackage{amstext}
\usepackage{bbm}
\usepackage{mathtools}
\usepackage{graphicx}
\usepackage{adjustbox}
\graphicspath{ {./plots/} }
\usepackage[mathscr]{euscript}
\usepackage{tabularx}
\usepackage{mathrsfs}
\usepackage{enumerate} 
\usepackage{cases} 
\usepackage{float} 
\usepackage{caption}
\usepackage{subcaption}
\usepackage{comment}
\usepackage{wrapfig}
\usepackage[utf8]{inputenc}
\usepackage[colorinlistoftodos]{todonotes}
\usepackage[T1]{fontenc}
\usepackage{amsbsy} 
\usepackage{physics} 
\usepackage{upgreek} 
\usepackage{diagbox} 
\usepackage{rotating}
\usepackage{algorithm}
\usepackage{algpseudocode}
\usepackage{layouts}

\theoremstyle{plain}

\newtheorem{theorem}{Theorem}[section]

\theoremstyle{definition}

\theoremstyle{remark}

\def \bsm {\boldsymbol}
\def \btheta {\boldsymbol{\theta}} 
\def \balpha {\boldsymbol{\alpha}}

\def \blambda {\boldsymbol{\lambda}}

\def \bbeta {\boldsymbol{\beta}}

\def \bq {\boldsymbol{\mathrm{q}}}

\def \bt {\boldsymbol{t}}

\def \Bernoulli {\mathrm {Bernoulli}}

\DeclareMathOperator*{\argmax}{arg\,max}

\begin{document}

\twocolumn[

\arxivtitle{ A Semiparametric Approach to Causal Inference }

\arxivauthor{ Archer Gong Zhang \And Nancy Reid \And Qiang Sun }

\arxivaddress{ University of Toronto \And University of Toronto \And University of Toronto } ]

\begin{abstract}
In causal inference, an important problem is to quantify the effects of interventions or treatments. Many studies focus on estimating the mean causal effects; however, these estimands may offer limited insight since two distributions can share the same mean yet exhibit significant differences. Examining the causal effects from a distributional perspective provides a more thorough understanding.
In this paper, we employ a semiparametric density ratio model (DRM) to characterize the counterfactual distributions, introducing a framework that assumes a latent structure shared by these distributions. Our model offers flexibility by avoiding strict parametric assumptions on the counterfactual distributions. 
Specifically, the DRM incorporates a nonparametric component that can be estimated through the method of empirical likelihood (EL), using the data from all the groups stemming from multiple interventions. 
Consequently, the EL-DRM framework enables inference of the counterfactual distribution functions and their functionals, facilitating direct and transparent causal inference from a distributional perspective.
Numerical studies on both synthetic and real-world data validate the effectiveness of our approach. 
\end{abstract}

\section{Introduction}
\label{intro}
Causal inference constitutes a foundational pursuit in scientific research, aiming to unveil the underlying cause-and-effect relationships between variables from their observed phenomena.  Across diverse disciplines such as medicine, social sciences, economics, and public policy, the ability to infer causal relationships is pivotal for informing decision-making, designing effective interventions, and advancing knowledge. 

A significant line of research in causal inference focuses on mean-centric estimands. Yet, the effect of a treatment or intervention often may not be adequately reflected by means. Traditional mean-centric causal effects estimands such as the average treatment effect (ATE) and conditional average treatment effect (CATE) may overlook important aspects of causal effects. As pointed out by \citet{kennedy2023semiparametric}, ``Causal effects are often characterized with averages, which can give an incomplete picture of the underlying counterfactual distributions.''

Given that causal effects are quantified through random variables by nature, it is more sensible to understand and study causal effects from a distributional viewpoint. 
Many recent studies investigate the causal effects from different angles than the mean estimation. 
For example, \citet{deuber2021estimation} study the difference in the extremely high or low quantiles, which offers a more suitable measure of causal effects for extreme events. 
\citet{lin2021causal} consider the situation where the potential outcomes take values in a space of distribution functions and propose new estimands of causal effects. 
\citet{chernozhukov2013inference} and \citet{kennedy2023semiparametric} explore inferences on distributions of the potential outcomes.

In this paper, we propose an approach to causal inference from a distributional perspective. 
We adopt a semiparametric density ratio model (DRM) \citep{anderson1979multivariate} for the counterfactual distributions corresponding to multiple interventions, also referred to as treatments.  In a nutshell, the DRM postulates that the ratios of the counterfactual densities share a common parametric form.  Two key implications of the DRM are: 1) it does not assume a parametric assumption for each counterfactual distribution, and 2) the common parametric form of the density ratios can be specified by the user. 
Thanks to its semiparametric structure, the DRM is flexible and the nonparametric component can be adaptively estimated from the data. 
The DRM has proven successful in various fields and applications 
\citep{qin1998inferences, fokianos2001semiparametric, de2017bayesian} including machine learning \citep{sugiyama2012density}.
Most importantly, the DRM effectively captures the similarity and connection between multiple distributions.
We find this particularly useful in causal inference because the counterfactual distributions stemming from multiple interventions often share common features and latent structures. 
At the same time, the parametric component helps model the shifts in the counterfactual distributions due to various interventions. 
Additionally, the DRM allows utilization of the entire dataset to estimate each distribution, resulting in higher efficiency compared to using only the dataset corresponding to each individual distribution.

\section{The problem setup}
\label{setup}

For each individual in a target population, let $X \in \mathbb{R}^p$ denote a covariate vector;
let $A$ denote the treatment assignment with multiple levels $1,\ldots,K$; let $Y(a) \in \mathbb{R}$ denote the potential outcome had the patient received the treatment $A = a$ \citep{rubin1974estimating}.
Let $\mathcal{P}(\cdot)$ represent the target population distribution, which is the underlying distribution of $(\{Y(a)\}, X, A)$. 

For each subject in the study, only one treatment is assigned, and thus we only observe the actual outcome $Y = Y(A)$. 
The actual data are given by $\{ (Y_i, X_i, A_i): i = 1, \ldots, n \}$, where $i$ denotes the index for a subject in the study. 
We assume that, conditional on the covariates $X_i$ and treatment $A_i$, the responses $\{ Y_i: i = 1, \ldots, n \}$ are independent and identically distributed (i.i.d.).

Our research problem is to efficiently estimate the causal effects, also known as the treatment effects, on the target population $\mathcal{P}$, which in principle measure the discrepancies among the distributions of the potential outcomes $\{ Y(a): a=1,\ldots,K \}$. We refer to the distribution of $Y(a)$ as a \textbf{counterfactual distribution} throughout the paper. 

Once the counterfactual distributions are estimated, one may naturally construct estimators of the causal effect estimands, such as ATE and CATE, as a byproduct. Under a binary treatment regime with $a=0,1$, the ATE is defined as:
$
\tau \coloneqq \mathbbm{E}_{\mathcal{P}}[Y(1) - Y(0)],
$
and the CATE is defined as:
$
\tau(x) \coloneqq \mathbbm{E}_{\mathcal{P}}[Y(1) - Y(0) \mid X = x].
$

\subsection{Assumptions for identifiable causal inference}\label{assumptions}

We list some assumptions in this paper that are standard in the causal inference literature.

\paragraph{Stable Unit Treatment Value Assumption (SUTVA)} 
Let $Y_i(a)$ be the potential outcome and $A_i$ the treatment assignment for unit $i \in \{1,\ldots,n \}$. Assume that 
1) No interference: $Y_i(A_i)=Y_i(A_1,\ldots,A_n)$, and 
2) Consistency: $Y(a)=Y$ if $A=a$. 

The first part of SUTVA requires that the potential outcomes for any unit do not interfere with the treatment assigned to other units (no interference), and the second part states that there are no different versions of each treatment level. 
This SUTVA assumption allows us to make inference about the distribution of $Y(a)$ for the individuals assigned to treatment $a$ using their observed $Y$. 

\paragraph{Unconfoundedness}
We assume that for our target population $\mathcal{P}$, conditional on the covariates $X$, the potential outcomes $\{ Y(a): a=1,\ldots,K \}$ are independent of the treatment assignment $A$:
\begin{align}\label{unconfound}
Y(a) \perp A | X, 
\quad \forall a=1,\ldots,K.
\end{align}
In other words, there is no unmeasured confounding that affects both the treatment assignment $A$ and potential outcomes $(Y(0),\ldots,Y(K))$. 
The most important implication of the unconfoundedness assumption is for the counterfactual distribution $Y(a)|X$ to be estimable from the observed data. 
In particular, $Y|X, A=a$ has the same distribution as $Y(a)|X$. 
This is essential because in reality, we only observe $Y=Y(A)$ but not the potential outcomes $(Y(0),\ldots,Y(K))$, and this assumption allows us to infer the unobserved $Y(a)$ for the individuals with $A \neq a$ using the observed $Y(a)$ for the individuals with $A=a$ that share the same covariates $X$.

\paragraph{Positivity} 
We also assume that each individual has a nonzero probability of being assigned to a treatment:  
$
P(A = a | X=x) \in (0,1)
$
for all $x$ and $a=1,\ldots,K$.

\section{A semiparametric density ratio model}
\label{sec:drm}

In this paper, we directly target estimating the counterfactual distribution of $Y(a)$, which allows a distribution-centric causal inference. 
To further enable assessing the heterogeneous treatment effect, which is useful in applications such as designing an individualized treatment regime, we consider the following semiparametric density ratio model \citep{anderson1979multivariate} for $Y(a)|X=x$. 

Denoting by $G(\cdot | x, a)$ the conditional counterfactual distribution of $Y(a) | X=x$, the DRM assumes that the conditional counterfactual distributions are connected and share some common latent structures: 
\begin{align}\label{DRM}
\mathrm{d}G(y | x, a) \!=\! \exp \{ \alpha(x,a) + \bbeta^{\top}(x,a) \bq(y) \} \mathrm{d}G_{0}(y), 
\end{align}
for some user-specified vector-valued basis function $\bq(y)$ with dimension $d$  and functional parameters $\{\alpha(x,a), \bbeta(x,a)\}$ that depend on the covariates $x$ and treatment $a$.
Here, $G_0(y)$ is an {\bf unknown} baseline distribution (also known as a reference distribution) function that only depends on $y$, serving as a common dominating measure for the counterfactual distributions $\{ G(\cdot | x, a): a \}$. 
We shall see later that $G_0(y)$ can be estimated nonparametrically using the method of empirical likelihood.

We also require the elements of $\bq(y)$ to be linearly independent. 
The linear independence is a natural assumption: otherwise, some elements of $\bq(y)$ are redundant. 
In addition, $\alpha(x,a)$ can be viewed as a functional normalizing constant uniquely defined as 
\begin{align}
\alpha(x,a) = -\log \int \exp \{ \bbeta^{\top}(x,a) \bq(y) \} \mathrm{d}G_{0}(y).
\label{alpha_def}
\end{align}

\subsection{DRM models distribution shift due to treatments}

Our DRM is applicable not only to multiple but also to continuous treatments, as evidenced by the flexibility in specifying the functional parameter $\bbeta(\cdot, a)$ over $a$.  In this paper, we focus on discrete treatments with multiple levels indexed by $a=1,\ldots,K$ for illustration purposes. 
This naturally induces $K$ subpopulation distributions $\{G(y | x, a): a=1,\ldots,K\}$ corresponding to the $K$ treatments. Consequently, the DRM in \eqref{DRM} can be equivalently written as: $\forall a = 1,\ldots,K$,
\begin{align}
\frac{\mathrm{d}G(y | x, a)}{\mathrm{d}G(y | x, 1)} = \exp \{ \bar \alpha_a(x) + \bar \bbeta_a^{\top}(x) \bq(y) \}, 
\label{DRM_binary}
\end{align}
where $\bar \alpha_a(x)=\alpha(x,a)-\alpha(x,1)$ and $\bar \bbeta_a(x)=\bbeta(x,a)-\bbeta(x,1)$. 
This reformulation clarifies that the DRM specifies a parametric form, as defined by the user, that models the distribution shifts from the counterfactual distributions of $Y(a)|X$ to $Y(1)|X$. 
Specifically, for each $x$, $\bar \bbeta_a(x)$ quantifies the magnitude of these shifts, while $\bq(y)$ indicates their common direction across all $K-1$ distribution shifts. 
Thus, the DRM framework explicitly delineates the causal effects of treatments or interventions from a distributional viewpoint. Additionally, our model maintains high flexibility by avoiding strict parametric assumptions on the counterfactual distributions, facilitated by the unspecified baseline distribution $G_0(y)$.

\subsection{DRM generalizes generalized linear models}

In fact, when $\bq(y)=y$, our DRM,  as specified in \eqref{DRM} can be regarded as a semiparametric generalization of a generalized linear model (GLM) within a natural exponential family: 
\[ 
f(y|x, a; \theta, \phi) = h(y, \phi) \exp \{ d(\phi) [y \theta - A(\theta)] \}, 
\]
where $\theta=g^{-1}((x^\top, a) \bbeta)$ with $g(\cdot)$ being the link function, and $\phi$ is the dispersion parameter assumed to be known. 
The DRM offers more flexibility than the GLM in two significant ways. 
First, the function $h(y,\phi)$, which corresponds to our unspecified baseline distribution $G_0(y)$, is assumed to be known in GLM but not in our DRM. 
Second, the basis function $\bq(y)$ is user-specified, allowing for a much richer class of functions than $y$ alone to handle various types of responses. 
Our model is closely related to the proportional likelihood ratio model in \citet{luo2012proportional} and the model by \citet{diao2012maximum}, with the basis function $\bq(y)=y$ and the DRM parameter $\bbeta(x)=x^\top \btheta$. 
\citet{huang2012proportional} and \citet{huang2014joint} also consider the proportional likelihood ratio model with $\bq(y)=y$, but their DRM parameter $\bbeta(x)$ is implicitly defined by a conditional mean model: 
$
\mathbb{E}(Y_k | X_k) = \eta (X_k^{\top} \btheta). 
$

To further illustrate this generalization, we now examine logistic regression. 
Consider a dichotomous outcome: $Y=0,1$, and let $(X, A)$ together form a vector of covariates. 
The standard logistic regression model assumes that $Y|X=x, A=a$ follows a Bernoulli distribution with parameter $p(x,a)=\mathbbm{P}(Y=1 | X=x, A=a)$, and further that  
$
\mathrm{logit}\{p(x,a)\} = (x^\top,a) \btheta.
$
This logistic regression model can be equivalently written as 
\begin{align}
& \mathbbm{P}(Y=y | X=x, A=a) 
= p(x, a)^y [1-p(x, a)]^{1-y} \nonumber \\ 
& = \exp \Big \{ -\log [1+\exp((x^\top,a) \btheta)] + (x^\top,a) \btheta y \Big \}, 
\label{LR}
\end{align}
for $y=0,1$.
If we disregard the baseline distribution $G_0(y)$, then the model in \eqref{LR} is equivalent to the DRM in \eqref{DRM} with the parametrization: 
$
\bq(y)=y, 
\bbeta(x,a)=(x^\top,a) \btheta,
$
and 
$
\alpha(x,a)=-\log [1+\exp((x^\top,a) \btheta)].
$

\subsection{Choices of $q(y)$ and $\beta(x,a)$}

The choice of the basis function $\bq(y)$ has been explored in the literature under a marginal DRM for response data only, from which we may draw some ideas. 
Based on exploratory data analysis, we could choose  $\bq(y)=(y, y^{2})^{\top}$ when the histograms of $Y$ appear  bell-shaped, 
or  $\bq(y)=(y, \log y)^{\top}$ when the histograms suggest gamma distributions, thus encompassing the normal or gamma distribution families within the DRM.
To ensure a sufficiently rich DRM, one might also consider $\bq(y)=(|y|^{1/2}, y, y^{2}, \log|y|)^{\top}$, covering distributions such as normal, gamma, binomial, or Poisson.
When there are more than two populations (or levels of treatments), the adaptive basis proposed by \citet{zhang2022density} can be used, where  $\bq(y)$  is determined by the leading functional principal components of the induced log density ratios.

In this paper, as a balance between model complexity and robustness, we prespecify the basis function $\bq(y)$ and delegate the inference of the DRM to  $\bbeta(x,a)$. 
Specifically, we allow a user-specified parametric form for $\bbeta(x,a)$ and estimate its parameters.  
A natural choice would be $\bbeta(x,a)=\btheta_a^\top x$, where $\btheta_a$ represents a $p\times d$ matrix of unknown parameters in $\mathbb{R}$. 
Alternatively, one might include higher-order terms of $x$ or employ other methods such as splines to model $\bbeta(x,a)$. Without a known parametric form, estimating the infinite-dimensional $\bbeta(x,a)$ for each $x$ and $a$ becomes challenging, particularly in the absence of repeated  $x$ values in the data.
Here we adopt the specification $\bbeta(x,a)$ as $\bbeta(x;\btheta_a)$  with a known parametric form, and we assume $\bbeta(x;\btheta_a)$ is a smooth function in $\btheta_a$.

An infinite-dimensional approach to the DRM is developed in \citet{izbicki2014high}, where they employ a tensor product to create a joint basis for $x$ and $y$, inspiring an alternative approach to model the exponential tilt in the DRM. They express  $\bbeta^\top(x,a) \bq(y)$ as a sum of products of the eigenbases derived from the marginal distributions of 
$X$ and $Y$, with unknown parameters. 
This configuration leads to a DRM  with infinite-dimensional $\bbeta(x,a)$ and $\bq(y)$. 
However,  it is more interpretable to consider scenarios where low-dimensional latent structures common to $G(y|x,a)$ exist, thereby affirming the DRM as a meaningful semiparametric model.

\subsection{Identifiability of the DRM}

    If we replace the baseline distribution $G_0$ with any exponentially tilted version of it with the same $\bq(y)$, the DRM remains the same.
    This introduces a problem with model identifiability concerning the DRM.
    To address this issue, we restrict  the choices of $G_0$ to those with 
    \begin{align}
    \mathbbm{E}_{G_0}[\bq(Y)] = \bsm{0}.
    \label {ref_const}
    \end{align}
    We show in the following theorem that such constrained choices of $G_0$ render the DRM identifiable.

\begin{theorem}
\label{DRM_identifiable}
For fixed values of $x$ and $a$, let $\alpha = \alpha(x,a)$, $\bbeta = \bbeta(x,a)$.
Let the family of distributions $ \big \{ \exp \{\alpha + \bbeta^{\top} \bq(y)\} \mathrm{d}G_{0}(y): (\bbeta, G_0) \big \}$ be the DRM in \eqref{DRM} indexed by an infinite-dimensional parameter $(\bbeta, G_0)$.
Further, assume in element-wise sense that  
\begin{align}
\int \bq(y) \bq^{\top}(y) \exp \{\alpha + \bbeta^{\top} \bq(y)\} \mathrm{d}G_{0}(y) < \infty. 
\label{finite-sec-moment}
\end{align}
Under condition \eqref{ref_const}, the DRM is identifiable in the sense that 
if $(\bbeta, G_0) \neq (\bbeta^{'}, G_0^{'})$ with both $G_0$ and $G_0^{'}$ satisfying \eqref{ref_const}, then 
\[
\exp \{\alpha + \bbeta^{\top} \bq(y)\} \mathrm{d}G_{0}(y) \neq \exp \{\alpha^{'} + \bbeta^{'\top} \bq(y)\} \mathrm{d}G_{0}^{'}(y).
\]
\end{theorem}
One may replace the zero-mean condition in \eqref{ref_const} by that $\mathbbm{E}_{G_0}[\bq(Y)]$ takes any specific finite value, and we assume without loss of generality that this value is zero. 
We also note that if the specification $\bbeta(x,a)=\bbeta(x;\btheta_a)$ satisfies
\begin{align}
\bbeta(x;\btheta_a)=\bbeta(x;\widetilde \btheta_a) 
\implies \btheta_a = \widetilde \btheta_a, 
\label{DRM_identifiable2}
\end{align}
then $\{\btheta_a:a=1,\ldots,K\}$ are also identifiable. The proof of this theorem is provided in the supplementary material.

\section{Empirical likelihood inference}

In the DRM framework \eqref{DRM}, the reference distribution $G_0(y)$ is not restricted and remains unspecified.
Assigning a parametric model to  $G_0(y)$
would reduce the DRM to a fully parametric model, which risks model misspecification. To avoid this risk, we employ a nonparametric method called empirical likelihood (EL) \citep{owen2001empirical}, which retains the effectiveness of likelihood methods without imposing restrictive parametric assumptions, making it an ideal platform for statistical inference under the DRM. EL's success in coupling with the DRM has been well-documented \citep {qin1993empirical, qin1997goodness, chen2013quantile, cai2017hypothesis}

In this section, we introduce an EL-based inference procedure to estimate the counterfactual distributions $G(y|x,a)$.
For presentation clarity, we categorize the observed data of the response and covariates into $K$ groups according to their treatment assignment, with data from the group assigned treatment $A=k$ denoted as $\{ (Y_{kj}, X_{kj}): j=1,\ldots,n_k \}$. 
We represent the total sample size as  $n=\sum_{k} n_k$ and introduce new notation for the counterfactual distributions as $G_k(y|x)=G(y|x,k)$ for $k=1,\ldots,K$. 
Thus we have 
\begin{equation}\label{data_mod}
Y_{kj} | X_{kj} \! \sim \! G_k(y|x) \!=\! \exp \{\alpha_k(x) \!+\! \bbeta^{\!\top\!}(x;\btheta_k) \bq(y)\} G_0(y), 
\end{equation}
for $j = 1,\ldots,n_k$, where $\alpha_k(x):=\alpha(x,k)$. 

We define the conditional EL for the model \eqref{DRM}, which is the probability of observing the data under a statistical model.
Let $p_{kj}=\mathrm{d}G_0(y_{kj})=\mathbbm{P}(Y=y_{kj}; G_0)$ denote the probability of observing $y_{kj}$ under the baseline distribution $G_0$ in \eqref{DRM}. 
Following \eqref{data_mod}, the conditional EL is given as
\begin{align}
& L(G_1,\ldots,G_k)
= \prod_{k,j} \mathrm{d}G_k(y_{kj} | x_{kj}) \nonumber \\ 
& = \prod_{k,j} p_{kj} \Big \{ \exp \{\alpha_k(x_{kj}) + \bbeta^\top(x_{kj};\btheta_k) \bq(y_{kj})\} \Big \}.
\label{EL}
\end{align}
For simplicity, we omit the index ranges in our expressions.
Notably, the EL is $0$ if $G_0$ is a continuous distribution; however, this does not undermine the utility of EL because any continuous distribution can be approximated well by a discrete distribution, allowing safe use of discrete distributions.

We regard the EL in \eqref{EL} as a function of the parameters $\btheta \coloneqq (\btheta_1, \ldots, \btheta_K)$ and $G_0$, and write its logarithm as
\begin{align}\label{log-EL}
& \ell (\btheta,G_0)
= \log L(G_1,\ldots,G_k) \nonumber \\ 
& = \sum_{k,j} \log p_{kj} \!+\! \sum_{k,j} \alpha_k (x_{kj}) \!+\! \bbeta^\top(x_{kj};\btheta_k) \bq(y_{kj}).
\end{align}
Following the pioneering work of \citet{owen2001empirical} and \citet{qin1994empirical}, we introduce the profile log-EL as a function $\btheta$ based on the log-EL in \eqref{log-EL} along with relevant constraints:
\begin{align} 
\label{profile-log-EL}
& \widetilde \ell (\btheta) 
= \sup_{G_{0}} \big \{ \ell (\btheta, G_0):
 \sum_{r,i} p_{ri} = 1,\, 
\sum_{r,i} p_{ri} \bq (y_{ri}) = \bsm{0}, 
\nonumber \\ 
& \sum_{r, i} p_{r i} \exp \{ \alpha_k(x_{kj}) + \bbeta^{\top}(x_{kj};\btheta_k) \bq(y_{ri}) \} = 1,\, \forall k,j
\big \}
\end{align} 
In the above, the first constraint ensures that $G_0$ is a valid distribution, the second constraint is the EL version of the model identifiability constraint in \eqref{ref_const},  and the third constraint is equivalent to the EL version of the normalizing constant constraint in \eqref{alpha_def} ensuring that $G_k (y|x_{kj})$ is a valid distribution.  These follow naturally by viewing $G_0$ as a discrete distribution with support $\{ y_{ri}: r,i \}$ and probability masses $\{ p_{ri}: r,i \}$.

Like the classical likelihood approach, we construct an estimator of $\btheta$ with the profile log-EL \eqref{profile-log-EL}, termed the maximum EL estimator (MELE): 
\begin{align}
\label {theta-MELE}
\widehat \btheta = \argmax_{\btheta} \widetilde \ell (\btheta).
\end{align}
Based on the MELE $\widehat \btheta$,  we can derive estimators of the counterfactual distributions $G(y|x,k)=G_k(y|x)$. 
Let $\widehat G_0(y)$ be the fitted cumulative distribution function (CDF) characterized by the estimators $\{\widehat p_{ri}\}$. 
Then, the EL-DRM estimator of $G_k(y|x)$ is
\begin{align} 
\label{cond_dist_est}
& \widehat G(y|x,k) \coloneqq \widehat G_k(y|x) \nonumber \\ 
\!=\! & \sum_{r,i} \widehat p_{ri} \exp \{ \widehat \alpha_k(x) \!+\! \bbeta^{\!\top\!}(x;\widehat \btheta_k) \bq(y_{ri}) \} \mathbbm{1}(y_{ri} \leq y), 
\end{align} 
where $\widehat \alpha_k(x)$ $\coloneqq -\log \sum_{r,i} p_{ri} \exp \{ \bbeta^\top(x;\widehat \btheta_k) \bq(y_{ri}) \}$ and $\mathbbm{1}(\cdot)$ is the indicator function. 
The estimated counterfactual distributions $\widehat G(y|x,k)$ are fundamental for many causal effects estimators such as those for ATE, CATE, and quantile treatment effects, which we explore through simulations and real-data analysis.

\subsection{The Lagrange multiplier method}

The optimization problem \eqref{profile-log-EL}  can be solved by using the Lagrange multiplier method.  Let $ \bt = \{ t_0, \{ t_{k j} \}_{k,j} \} $ and $ \blambda $ be Lagrange multipliers, where $ \bt $ is a vector of length $ n+1 $ and $ \blambda $ has the same dimension as $ \bq (y) $. 
Furthermore, we treat $\alpha_k(x_{kj})$ as a free parameter, denoted by $\alpha_{kj}$, and let $\balpha=\{\alpha_{kj}: k,j\}$ represent the collection of all $\alpha_{kj}$. 
We shall see later that the dependency of $\alpha_k(x_{kj})$ on $\btheta_k$ and $G_0$ defined by its role as a normalizing constant in \eqref{alpha_def} is incorporated in a constraint.
We define the Lagrangian as:
\begin{align*} 
& \mathcal{L}(\btheta, G_0, \balpha, \bt, \blambda) \nonumber \\ 
& = \ell(\btheta, G_0) + t_0 \big \{ 1-\sum_{r,i} p_{ri} \big \} 
+ \blambda^{\top} \big \{ \bsm{0}-\sum_{r,i} p_{ri} \bq(y_{ri}) \big \} \nonumber \\ 
& \!\!+\! \sum_{k,j} t_{k j} \big \{ 1 \!-\! \sum_{r, i} p_{r i} \exp \{ \alpha_k (x_{kj}) \!+\! \bbeta^{\!\top\!}(x_{kj};\btheta_k) \bq(y_{ri}) \} \big \}.
\end{align*} 
According to the Karush--Kuhn--Tucker theorem \citep{boyd2004convex}, the solution to the constrained optimization problem in \eqref{profile-log-EL} satisfies  that for any $ \btheta $,  $\nabla_{p_{ri}, \balpha, \bt, \blambda} \mathcal{L}(\btheta, G_0, \balpha, \bt, \blambda)=\bsm{0}$.
Let $(\widehat p_{ri}, \widehat \balpha, \widehat \bt, \widehat \blambda)$ denote the solution, where some calculations show  $\widehat t_{kj} = 1, \widehat t_0 = 0$.
Moreover, 
\begin{equation}
\widehat p_{ri}^{\,-1} \!=\!
\! \sum_{k, j} \! \exp \{ \widehat \alpha_{k j} + \bbeta^{\!\top\!}(x_{kj};\btheta_k) \bq(y_{ri}) \}
+ \widehat \blambda^{\!\top\!} \bq(y_{ri}),
\label{eqn_p}
\end{equation} 
where $(\widehat \balpha, \widehat \blambda)$ satisfies
\begin{align} 
& 1 - \sum_{r,i} \widehat p_{ri} \exp \{ \widehat \balpha_{k j} + \bbeta^\top(x_{kj};\btheta_k) \bq (y_{ri}) \} = \bsm{0}, \nonumber \\ 
& \sum_{r, i} \widehat p_{r i} \bq(y_{ri}) = \bsm{0}.
\label{eqn_alpha}
\end{align}

\subsection {Algorithms to solve for $(\widehat p_{ri}, \widehat \balpha)$}

A key observation is that $ \widehat p_{r i} $ and $ \widehat \balpha $ implicitly depend on each other. 
Therefore, although it seems we have their closed-form solutions, this is not actually the case. 
One method to resolve this is by using an iterative algorithm in Algorithm \ref{alg1}. 
However, this method often requires many iterations to converge, especially since the dimension of $\widehat \balpha$ is as large as the sample size $n$. 

To ease the computation, another approach is to obtain a sensible approximation of $\widehat p_{r i}$ by using only the response data --  ignore the covariates $x$ -- through a marginal DRM fitted to $\{Y_i\}$.  
In this step, we also find a solution for $\blambda$.
Subsequently, we can easily compute $\widehat \balpha$ based on \eqref{eqn_alpha} because now $\widehat p_{r i}$ no longer depends on those $\widehat \alpha_{k j} = \widehat \alpha_k (x_{k j})$ values. 
Perhaps surprisingly, ignoring the covariates $x$ in the step of finding $\widehat p_{ri}$ does not  lose much information. 
To see this, we examine the expression \eqref{eqn_p} for $\widehat p_{ri}$, focussing on the exponential function part, and have:
\begin{align} 
& \sum_{k,j} \exp \{ \alpha_{k j} + x_{k j}^{\top} \btheta_k \bq (y_{r i}) \} \nonumber \\ 
& \! \approx \! \sum_{k} n_k \sum_{j = 1}^{n_k} \widehat{\mathbbm{P}}(X = x_{kj}; G_k) \frac {\mathrm{d}G_k(y_{ri} | x_{kj})} {\mathrm{d}G_0(y_{ri})} \nonumber \\ 
& \! \approx \! \sum_{k} n_k \frac {\mathrm{d} \widetilde G_k(y_{ri})} {\mathrm{d}G_0(y_{ri})} 
\! \approx \! \sum_{k} n_k \exp \{ \widetilde \balpha_{k} + \widetilde \bbeta_k^{\top} \bq (y_{r i}) \},
\label{approx_drm}
\end{align} 
where we use $\widehat{\mathbbm{P}}(X = x_{kj}; G_k)=1/n_k$ in the first approximation, 
law of total probability in the second approximation with $\widetilde G_k$ denoting the marginal distribution, 
and a $Y$-marginal DRM to approximate the density ratio $\widetilde G_k/G_0$ for the marginal distributions in the third approximation.
Now replacing the exponential part in \eqref{eqn_p} for $\widehat p_{ri}$ with the expression in \eqref{approx_drm} yields exactly the solution to $p_{r i}$ under a $Y$-marginal DRM. 
Our simulations suggest that this approximation method gives very similar results to the iterative algorithm \ref{alg1} in terms of estimating the causal effects. 
This approximate method of finding $(\widehat p_{ri}, \widehat \balpha)$ is summarized in Algorithm \ref{alg2}. 

\begin{algorithm}[t!]
\caption{Solve for $(\widehat p_{ri}, \widehat \balpha)$}
\label{alg1}
\begin{algorithmic}
\State {\bf Initialization}: $\{\widehat p_{ri}^{(0)}\}, \widehat \balpha^{(0)}, \widehat \blambda^{(0)}$
\State {\bf Previous iteration}: $\{\widehat p_{ri}^{(t-1)}\}, \widehat \balpha^{(t-1)}, \widehat \blambda^{(t-1)}$
\Repeat
\State {\bf Update $\widehat p_{ri}$}: obtain \\ 
\begin{center}
$
\widehat p_{ri}^{(t)} = \big \{ \sum_{k,j} \exp \{ \widehat \alpha_{kj}^{(t-1)} + \bbeta^\top(x_{kj};\btheta_k) \bq(y_{ri}) \} 
+ \widehat \blambda^{(t-1)\top} \bq(y_{ri}) \big \}^{-1}
$  
\end{center}
\State {\bf Update $\widehat \balpha$}: obtain \\ 
\begin{center}
 $\widehat \alpha_{kj}^{(t)}\!=\!-\log \sum_{r,i} \widehat p_{ri}^{(t)} \exp \{ \bbeta^\top(x_{kj};\btheta_k) \bq(y_{ri}) \}$ 
\end{center}
\State {\bf Update $\widehat \blambda$}: obtain $\widehat \blambda^{(t)}$ by solving \\ 
\begin{center}
$\sum_{r,i} \big \{ \sum_{k,j} \exp \{ \widehat \alpha_{kj}^{(t)} + \bbeta^\top(x_{kj};\btheta_k) \bq(y_{ri}) \}
+ \widehat \blambda^{\top} \bq(y_{ri}) \big \}^{-1} \bq(y_{ri}) = \bsm{0}$ 
\end{center}
\Until 
$\{\widehat p_{ri}^{(t)}\}, \widehat \balpha^{(t)}, \widehat \blambda^{(t)}$ 
converge
\end{algorithmic}
\end{algorithm}

\begin{algorithm}[t!]
\caption{Approximation for $(\widehat p_{ri}, \widehat \balpha)$}
\label{alg2}
\begin{algorithmic}
\State {\bf Input}: the response data $\{Y_{kj}\}$ and a $Y$-marginal DRM
\State {\bf Step 1}: fit the marginal DRM to the response data using the method of EL
\State {\bf Step 2}: obtain the EL-based estimators $\widehat p_{ri}^{\,\mathrm{mar}}$ of $\mathrm{d}G_0(y_{ri})$ under the $Y$-marginal DRM
\State {\bf Step 3}: $\widehat \alpha_{kj}^\mathrm{mar}\!=\!-\log \sum_{r,i} \widehat p_{ri}^{\,\mathrm{mar}} \!\exp \{ \bbeta^{\!\top}\!(x_{kj};\btheta_k) \bq(y_{ri}) \}$
\State {\bf Return}: $\{\widehat p_{ri}^{\,\mathrm{mar}}\}$ and $\{\widehat \alpha_{kj}^\mathrm{mar}\}$ as approximations of $\{\widehat p_{ri}\}$ and $\{\widehat \alpha_{kj}\}$
\end{algorithmic}
\end{algorithm}

\subsection{Theoretical properties} 

After finding $\widehat \balpha$, $\widehat p_{ri}$, and $\widehat \blambda$, we obtain the analytical expression of the profile log-EL defined in \eqref {profile-log-EL}: 
\begin{align*}
& \widetilde \ell(\btheta) 
= 
\sum_{r,i} \{ \widehat \alpha_{r i} + \bbeta^\top(x_{ri};\btheta_r) \bq(y_{ri}) \} - \\
&
\sum_{r,i} \log 
\sum_{k,j} \exp \{ \widehat \alpha_{kj} + \bbeta^\top(x_{kj};\btheta_k) \bq(y_{ri}) \}
+ \widehat \blambda^{\top} \bq (y_{ri}).
\label{profile-log-EL-explicit} 
\end{align*}
If $\widehat \btheta$ is in the interior of the parameter space, we can find the MELE $\widehat \btheta$, defined in \eqref{theta-MELE}, as the solution to the score equation $\nabla \widetilde \ell(\btheta) = \bsm{0}$. 
After some algebra, we can also find the analytical expression of the score function $\nabla \widetilde \ell(\btheta)$.
For notational simplicity, we let
$
\widehat{\mathbbm{E}}_{k}[\bq(Y) | x_{k j}] =
\sum_{r,i} \bq(y_{ri}) \widehat p_{ri} \exp \{ \widehat \alpha_{k j} + \bbeta^{\top}(x_{kj};\btheta_k) \bq(y_{ri}) \}
$
denote the first moment of $\bq(Y)$ with respect to the estimated conditional distribution 
$
\widehat G_{k}(y | x_{k j})
= 
\sum_{r,i} \widehat p_{r i} \exp \{ \widehat \alpha_{k j} + \bbeta^{\top}(x_{kj};\widehat \btheta_k) \bq(y_{ri}) \} \mathbbm {1} (y_{ri} \leq y).
$
Then, we have that 
\begin{equation*}
\frac {\partial \widetilde \ell (\btheta)} {\partial \btheta_{k}} = 
\sum_{j = 1}^{n_{k}} \nabla_{\btheta_k}^\top \bbeta(x_{kj};\btheta_k)
\left \{ \bq (y_{k j}) - \widehat {\mathbbm{E}}_{k}[\bq(Y) | x_{k j}] \right \}.
\label{score-fun-explicit}
\end{equation*}
The following theorem states that the MELE $\widehat \btheta$ in \eqref{theta-MELE} and the EL-based estimator of the baseline distribution $\widehat G_0(y) \coloneqq \sum_{r,i} \widehat p_{ri} \mathbbm{1}(y_{ri} \leq y)$ are consistent under some conditions imposed solely for technical reasons. 
The proof is provided in the supplementary material.
\begin{theorem}
\label{consistency}
Under the assumptions in Section~\ref{sec:drm} and the conditions that: 
(i) the covariate space $\mathcal{X}$ for $x$ is bounded, the response space $\mathcal{Y}$ for $y$ is compact, and the parameter space $\mathcal{B}$ for $\btheta_k$ is compact; 
(ii) $\bbeta(x; \btheta_k)$ satisfies \eqref{DRM_identifiable2}; and 
(iii) $\{\bbeta(x; \btheta_k): \btheta_k \in \mathcal{B} \}$ is a uniformly bounded and Glivenko--Cantelli class
of functions on $\mathcal{X}$, 
then $\widehat \btheta$ is consistent and $\widehat G_0(y)$ converges uniformly to the truth $G_0^*(y)$.
\end{theorem}
We remark that continuous $\bbeta(x; \btheta_k)$ on the product space $\mathcal{X} \times \mathcal{B}$ with compact $\mathcal{X}$ and $\mathcal{B}$ satisfies condition (iii) in Theorem~\ref{consistency}. 
Our example $\bbeta(x; \btheta_k)=\btheta_k^\top x$ is an instance.

\section{Numerical studies}

This section evaluates the effectiveness of our proposed method using synthetic experiments and a real-data analysis.

\subsection{Causal effect estimands and estimators}

For illustration, we focus on the ATE, CATE, and quantile treatment effects on the treated (QTET) in our numerical studies. 
For a binary treatment, the QTET at quantile level $p \in (0, 1)$ is defined as 
\[
\mathrm{QTET}(p) = F_{Y(1)|A=1}^{-1}(p)-F_{Y(0)|A=1}^{-1}(p), 
\]
where $F_{Y(a)|A=1}^{-1}(p)$ is the quantile function for $Y(a)|A=1$. 
The EL-DRM estimators for these estimands are constructed using the estimated counterfactual distributions $\widehat G(y|x,a)$ in \eqref{cond_dist_est}. 

We also compare our method with some competing methods in our simulations. For ATE estimation, we include the classical regression-based G-formula, logistic model-based Inverse Probability Weighting (IPW), and Augmented Inverse Probability Weighted (AIPW) estimators. 
For QTET estimation, we include 1) the IPW-based quantile estimators by \citet{firpo2007efficient}, 
and 2) methods by \citet{chernozhukov2013inference} for  estimating counterfactual distributions using quantile regression (CF(QR)) and using distribution regression under a logistic model (CF(logit)). 
The {\tt R} packages for implementing these QTET methods are available as: {\tt qte} \citep{qte_package} and {\tt Counterfactual} \citep{counterfactual_package}. 

\subsection{Simulation studies}

\subsubsection{Data generated from Gaussian distribution}

We first generate synthetic data from Gaussian distributions according to the following data generating process:
\begin{gather*} 
A \!\sim\! \Bernoulli (0.5), 
X_1 \!\sim\! N(1,1), 
X_2 | X_1, A \!\sim\! N(2AX_1, 1),
\\ 
Y \!=\! 1 \!+\! A \!+\! X_1 \!+\! 2 A X_1 \!+\! 0.5 A X_1^2 \!+\! AX_2 \!+\! \varepsilon, 
\, \,
\varepsilon \!\sim\! N (0, 1).
\end{gather*} 
The true CATE is $\tau(x)=1+2x_1-0.5x_1^2+x_2$, and the true ATE is $\tau=3$. 
The true QTET at the five selected quantile levels $p=0.1, 0.3, 0.5, 0.7, 0.9$ are $(0.091, 2.847, 4.444,$ $5.792, 7.328)$.

We applied the DRM to simulated data for the counterfactual distributions $G(y | x_1, x_2, a)$ with $\bq(y)=(y, y^2)^\top$ and $\bbeta^\top_{\mathrm{cor}}(x,a)=(x_1, x_1^2, x_2) \btheta_a$ to ensure the model is correctly specified (``full''). 
To assess the impact of potential model misspecification, we also use $\bbeta^\top_{\mathrm{mis1}}(x,a)=(x_1, x_2) \btheta_a$ (``mis1'') and $\bbeta^\top_{\mathrm{mis2}}(x,a)=x_1 \btheta_a$ (``mis2''). 
The performance of QTET and ATE estimators are summarized in Figure~\ref{S1_boxplot}; detailed numerical results are in the supplementary material. 
The results are obtained based on 1000 simulation repetitions with a sample size of $n=1000$. 
Each simulation repetition takes less than a minute (about 10 seconds for DRM) to run all the methods on an Apple M1 Pro laptop with 32GB memory and 512GB storage.

We observe that the DRM estimator of ATE performs comparably to the G-formula estimator with the correct regression model (``full''), and significantly better than the IPW estimator due to the incorrectly specified logistic regression model for the propensity score in IPW. Furthermore, the DRM estimators remain effective even when the model is mildly misspecified by excluding the quadratic term $x_1^2$ from $\bbeta(x,a)$. In contrast, the AIPW method performs poorly when $x_1^2$ is excluded from the model fitting (``mis1''). 
As expected, all the methods lead to biased estimators when only $x_1$ is considered in the analysis (``mis2'').
For QTET estimation, the DRM method has much smaller bias than competing methods. The standard errors are slightly larger than the CF(QR), which is due to their semiparametric nature.  Interestingly, in all QTET estimation methods, omitting some covariates does not degrade performance. The estimated CATE $\tau(x)$ from one simulation run is illustrated in Figure~\ref{S1_CATE}, showing that both the correctly specified and mildly misspecified DRM provide a satisfactory depiction of the true CATE. 

\begin{figure}
\centering 
\includegraphics[width = \textwidth/2, keepaspectratio]{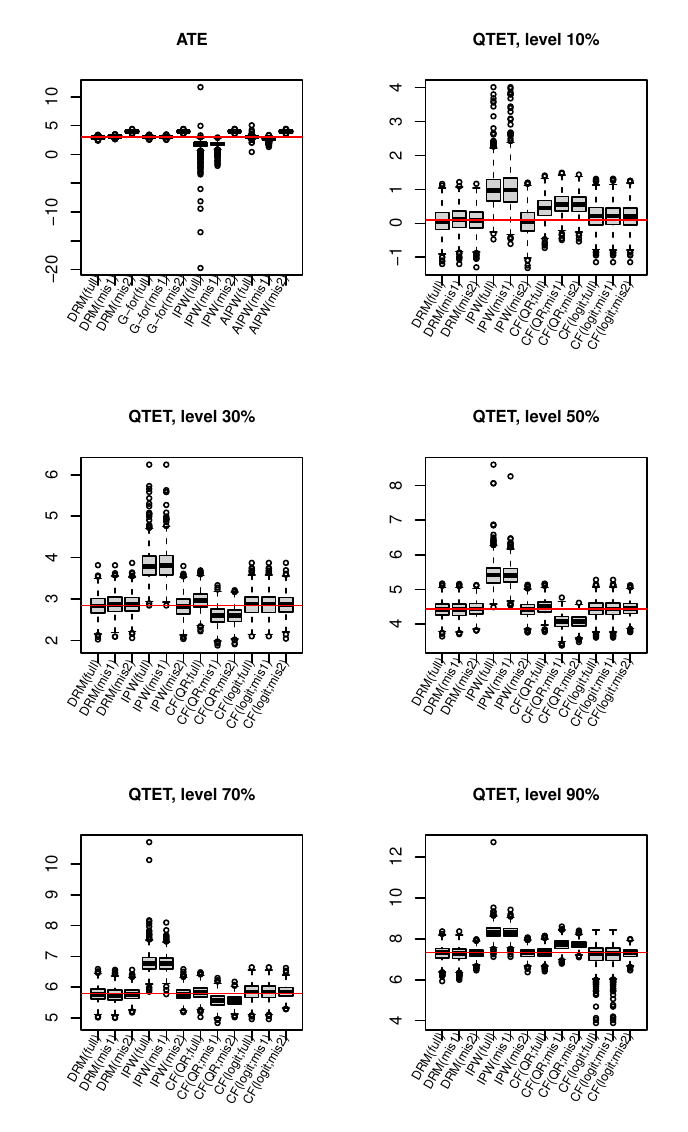}
\caption{Boxplots of the ATE and QTET estimators using various methods based on 1000 simulation repetitions of Gaussian data. 
The red lines depict the truth.}
\label{S1_boxplot}
\end{figure}

\begin{figure}
\centering 
\includegraphics[width = \textwidth*3/13, keepaspectratio]{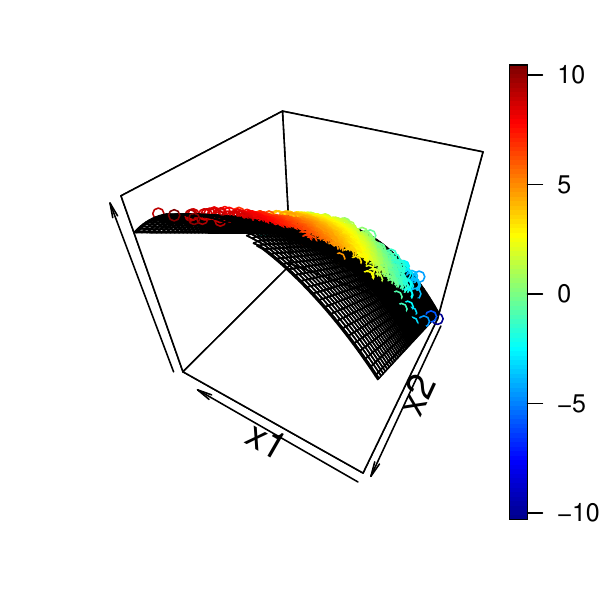}
\includegraphics[width = \textwidth*3/13, keepaspectratio]{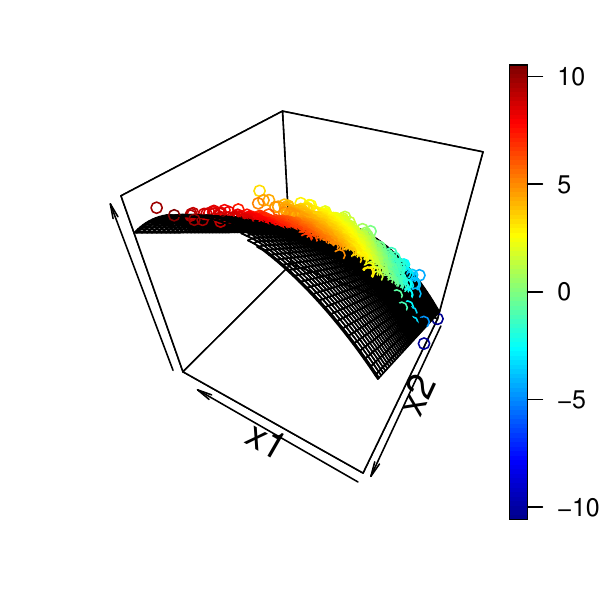} 
\caption{CATE plots based on one simulated Gaussian data. 
Use DRM with $\bq(y)=(y, y^2)^\top$.
Left: $\bbeta_{\mathrm{cor}}(x,a)$.
Right: $\bbeta_{\mathrm{mis1}}(x,a)$. 
The black surface depicts the truth.}
\label{S1_CATE}
\end{figure}

\subsubsection{Data generated from gamma distribution}

To simulate scenarios that don't exhibit Gaussianity, we also conduct experiments for data from gamma distributions:
\begin{gather*} 
A \sim \Bernoulli (0.5), 
\quad  
X_1 \sim \mathrm{Gamma}(1,0.5), \\ 
X_2 | X_1, A \!\sim\! \mathrm{Gamma}((A+1)(X_1+1), 1), \\ 
\mu(x_1, a) \!\sim\! \mathrm{Gamma}((a+1)(x_1+1), 1),
\\ 
Y = 0.5 (A+1) \mu (X_1, A) + 0.5 (A+1) X_2.
\end{gather*} 
The true CATE is $\tau(x)=1.5(1+x_1)+0.5x_2$, and the true ATE is $\tau=3.375$. 
The true QTET at levels $p=0.1, 0.3, 0.5, 0.7, 0.9$ are (1.730, 2.619, 3.411, 4.384, 6.188).

We fit the DRM to the data both with the correct specification: $\bq_{\mathrm{cor}}(y)=(y, \log y)^\top$ and $\bbeta^\top_{\mathrm{cor}}(x,a)=(x_1, x_2) \btheta_a$ (``full''), and with mild model misspecification: $\bq_{\mathrm{mis}}(y)=(y, y^2)^\top$ (``mis1'') and $\bbeta^\top_{\mathrm{mis}}(x,a)=x_1 \btheta_a$ (``mis2''). 
Correspondingly for the other competing methods, we also consider both using the full set of covariates $(x_1,x_2)$ (``full'') and using only $x_1$ (``mis'').
The results are illustrated in Figure~\ref{S2_boxplot}, with CATE plots and numerical results provided in the supplementary material, from which we observe similar phenomena as the Gaussian situation. 
We also provide some additional simulation results for other distributions in the supplementary material.

\begin{figure}
\centering 
\includegraphics[width = \textwidth/2, keepaspectratio]{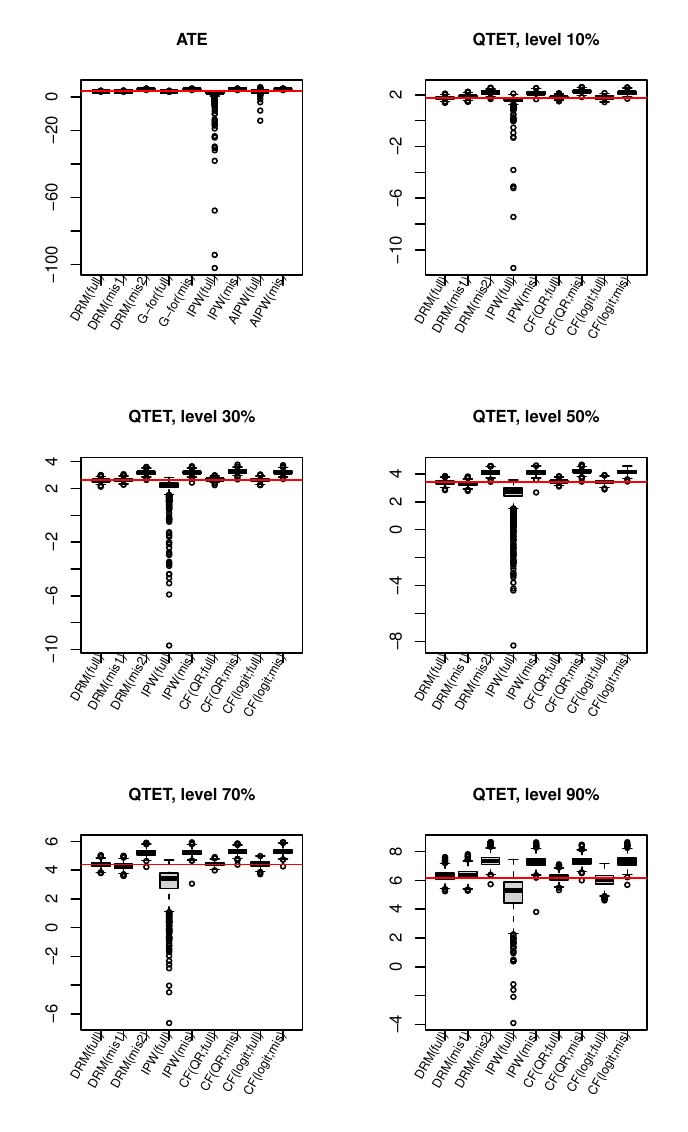}
\caption{Boxplots of the ATE and QTET estimators using various methods based on 1000 simulation repetitions of gamma data. 
The red lines depict the truth.}
\label{S2_boxplot}
\end{figure}

\begin{figure}
\centering 
\includegraphics[width = \textwidth*4/11, keepaspectratio]{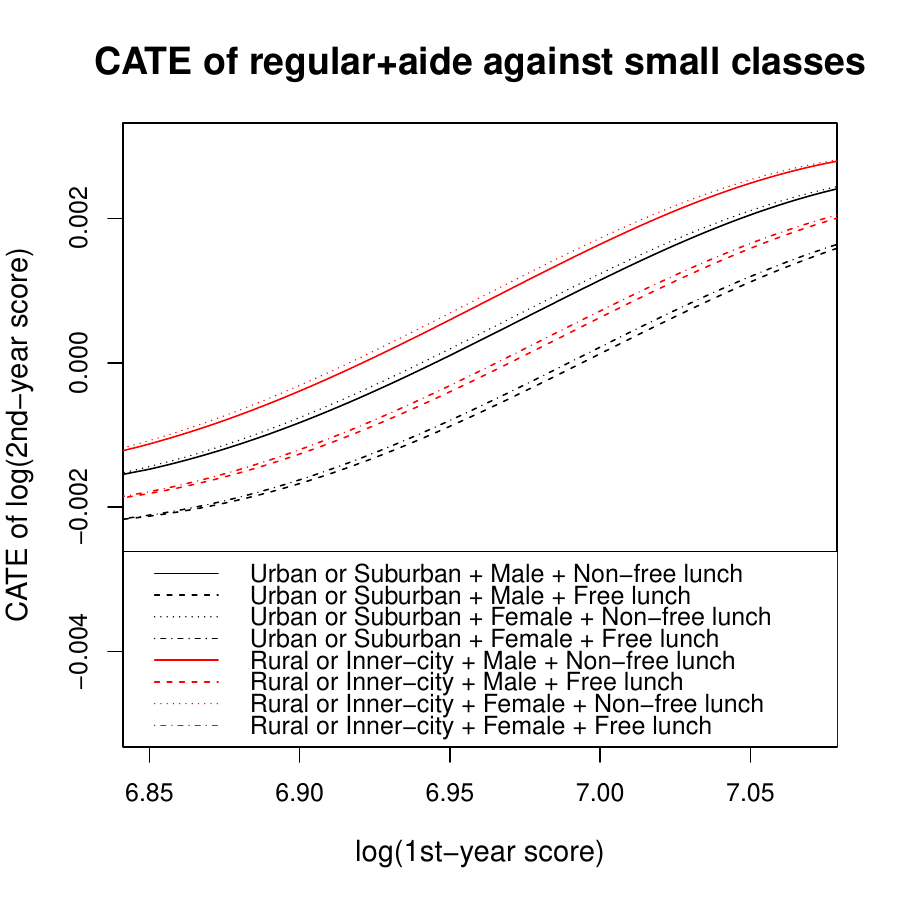}
\caption{Plot of the estimated CATE for the eight groups of students in STAR real data.}
\label{real_CATE}
\end{figure}

\subsection{Real-data analysis: STAR project}

We also analyze a real data from the Tennessee Student Teacher Achievement Ratio (STAR) project \citep{DVN/SIWH9F_2008}. 
The relevant data used in our analysis is also accessible from the {\tt R} package {\tt AER} \citep{AER_package}. 
The STAR project was a randomized experiment that studied the effect of class size on students' test scores,  monitoring students from kindergarten through third grade. 
Students were randomly assigned to one of three interventions: small class (13 to 17 students per teacher), regular class (22 to 25 students per teacher), and regular-with-aide class (22 to 25 students with a full-time teacher's aide).

For the analysis, we treat the class size in second grade as the treatment variable, with three levels, and use the log-transformed sum of second-grade reading and math test scores as the outcome variable to investigate the causal effect of class size. Additionally, our covariates include the log-transformed test score from first grade, gender (female or male), qualification for free lunch (free or non-free), and school area (rural or inner-city versus urban or suburban). After removing missing values, our dataset comprises 4370 students: 1426 in small class, 1412 in regular class, and 1532 in regular-with-aide class.
We apply the DRM with $\bq(y)=y$ and $\bbeta^\top(x,a)=x^\top \btheta_a$ to this data. 
The estimated CATE for regular-with-aide versus small classes is shown in Figure~\ref{real_CATE}. 
The CATEs appear to cluster differently based on the demographic factors captured by the covariates, which highlights the heterogeneous impact of class size on student outcomes.
We also show the plot for the estimated counterfactual distributions in the supplementary material.

\section{Discussion}

This paper presents a semiparametric approach to causal inference based on a DRM, which characterizes causal effects from a distributional perspective. The proposed DRM comprises two parts: a parametric part capturing common latent structures across all counterfactual distributions, and a nonparametric part representing the unspecified baseline distribution $G_0$. We use empirical likelihood to adaptively estimate $G_0$ from the observed data and propose an approximation algorithm to facilitate computation.
In situations where researchers are uncertain about the underlying family of the counterfactual distributions or seek to enhance robustness by avoiding strict parametric assumptions, the DRM offers a viable option. 
We want to remark that the objective of this paper is not to achieve state-of-the-art performance, but rather to introduce a framework enabling robust, statistically interpretable, and distribution-centric causal inference.

Due to the discrete nature of empirical likelihood, the proposed estimators of counterfactual distributions are inherently discrete. This limitation may constrain the applicability of our method if the desired causal effect estimand relies on smoothness measures of the underlying counterfactual distributions.  Another limitation is on the specification of $\bq(y)$ and $\bbeta(x,a)$. Although we offered some insights into these specifications through heuristic approaches, there is a need to develop methods for the data-adaptive selection of these parameters with both interpretability and theoretical guarantees.
Furthermore, while our simulations explore the ramifications of potential model misspecification, we believe a comprehensive study of the robustness of the DRM would be valuable for the community. We leave these tasks to future research.

\bibliographystyle{apalike}
\bibliography{Archer_bib2024}

\setcounter{section}{0}
\onecolumn
\arxivtitle{Supplementary Materials}



\section{Proofs}
\label{supp_proofs}

\begin{proof}[Proof of Theorem~\ref{DRM_identifiable}]
We prove this by contradiction. 
Without loss of generality, assume $\bbeta^{'}=\bsm{0}$ and hence $\alpha^{'}=0$. 
Suppose
\begin{align}
\exp \{\alpha_0 + \bbeta_0^{\top} \bq(y)\} \mathrm{d}G_{0}(y) = \mathrm{d}G_{0}^{'}(y), 
\label{contradiction_assumption}
\end{align}
we want to show that $(\bbeta_0, G_0)=(\bsm{0}, G_0^{'})$. 
If we can show that $\bbeta_0=\bsm{0}$, then the proof is complete. 

Recall that with a known $G_0$, the normalizing constant $\alpha$ is a function of $\bbeta$: 
\[
\alpha = \alpha(\bbeta) = -\log \int \exp \{ \bbeta^\top \bq(y) \} \mathrm{d}G_{0}(y).
\]
Also, recall that we restrict the choices of $G_0$ to those with 
\begin{align}
\mathbbm{E}_{G_0}[\bq(Y)] = \bsm{0}.
\label {ref_const_supp}
\end{align}

We define a function 
\[
\bsm{H}(\bbeta)= \int \bq(y) \exp \{\alpha + \bbeta^{\top} \bq(y)\} \mathrm{d}G_{0}(y).
\]
By assumptions \eqref{contradiction_assumption} and \eqref{ref_const_supp}, we have 
\[
\bsm{H}(\bbeta_0)= \int \bq(y) \exp \{\alpha_0 + \bbeta_0^{\top} \bq(y)\} \mathrm{d}G_{0}(y)
= \int \bq(y) \mathrm{d}G_{0}^{'}(y) = \bsm{0}.
\]
Also observe that by assumption \eqref{ref_const_supp}, 
\[
\bsm{H}(\bsm{0})= \int \bq(y) \mathrm{d}G_{0}(y) = \bsm{0}.
\]
Therefore, if we can show that $\bsm{H}(\bbeta)$ is injective, then we must have $\bbeta_0=\bsm{0}$. 
We show the injectivity by proving that the Jacobian of $\bsm{H}(\bbeta)$ is positive definite. 
Note that 
\[
\nabla\alpha(\bbeta) = -\int \bq(y) \exp \{\alpha + \bbeta^{\top} \bq(y)\} \mathrm{d}G_{0}(y), 
\]
and we have 
\begin{align*}
\nabla \bsm{H}(\bbeta) 
= & \int \bq(y) [\bq(y) + \nabla\alpha(\bbeta)]^\top \exp \{\alpha + \bbeta^{\top} \bq(y)\} \mathrm{d}G_{0}(y) \\ 
= & \int \bq(y) \bq^\top(y) \exp \{\alpha + \bbeta^{\top} \bq(y)\} \mathrm{d}G_{0}(y) \\ 
& - \left[ \int \bq(y) \exp \{\alpha + \bbeta^{\top} \bq(y)\} \mathrm{d}G_{0}(y) \right] \left[ \int \bq(y) \exp \{\alpha + \bbeta^{\top} \bq(y)\} \mathrm{d}G_{0}(y) \right]^\top, 
\end{align*}
which can be seen as a covariance matrix and is therefore positive definite because elements of $\bq(y)$ are assumed to be linearly independent.

This completes the proof.
\end{proof}

\begin{proof}[Proof of Theorem~\ref{consistency}]

We build on the proof strategy of Theorem 2.2 in \citet{murphy1997maximum} using empirical process theory.
Recall that 
\[
\mathrm{d}\widehat G_0(y_{ri})
=\widehat p_{ri} = 
\Big \{ 
\sum_{k, j} \exp \{ \widehat \alpha_{k j} + \bbeta^\top(x_{kj};\widehat \btheta_k) \bq(y_{ri}) \}
+ \widehat \blambda^{\top} \bq(y_{ri})
\Big \}^{-1}.
\]
By replacing $(\widehat \balpha,\widehat \btheta)$ by their true values $(\balpha^*, \btheta^*)$ and $\widehat \blambda$ by $\bsm{0}$, we define 
\[
\widetilde p_{ri} = 
\Big \{ 
\sum_{k,j} \exp \{ \alpha_{kj}^* + \bbeta^\top(x_{kj};\btheta_k^*) \bq(y_{ri}) \}
\Big \}^{-1}, 
\]
and define correspondingly its induced CDF: $\widetilde G_0(y) \coloneqq \sum_{r,i} \widetilde p_{ri} \mathbbm{1}(y_{ri}\leq y)$.

Since $\{\bbeta(x; \btheta_k): \btheta_k \in \mathcal{B} \}$ is a uniformly bounded Glivenko--Cantelli class, by the Glivenko--Cantelli Theorem \citep{van1996weak}, 
\[
\widetilde G_0(y) - \mathbbm{E}\Big[\sum_{r,i} \widetilde p_{ri} \mathbbm{1}(y_{ri}\leq y)\Big]
\to 0
\]
uniformly for $y \in \mathcal{Y}$ as $n \to \infty$. 

We next work out the expression for $\mathbbm{E}[\sum_{r,i} \widetilde p_{ri} \mathbbm{1}(y_{ri}\leq y)]$.  
We first observe that $\widetilde p_{ri}$ can be expressed by
\[
\widetilde p_{ri} = \frac{\mathrm{d}G_0^*(y_{ri})}{\sum_{k,j}\mathrm{d}G_k^*(y_{ri} | x_{kj})}, 
\]
where $G_0^*$ and $G_k^*$ are the true CDF for $G_0$ and $G_k$.
Regarding $y_{ri}$ as a sample from the conditional distribution $G_r^*(\cdot|x_{ri})$, 
we then have  
\begin{align*} 
\mathbbm{E}\Big[\sum_{r,i} \widetilde p_{ri} \mathbbm{1}(y_{ri}\leq y)\Big]
& = \sum_{r,i} \int \frac{\mathrm{d}G_0^*(t)}{\sum_{k,j}\mathrm{d}G_k^*(t | x_{kj})} \mathbbm{1}(t\leq y) \mathrm{d}G_r^*(t|x_{ri})  \\ 
& = \int \frac{\sum_{r,i} \mathrm{d}G_r^*(t|x_{ri})}{\sum_{k,j}\mathrm{d}G_k^*(t | x_{kj})} \mathbbm{1}(t\leq y) \mathrm{d}G_0^*(t) \\ 
& = G_0^*(y). 
\end{align*}
As a result, $\widetilde G_0(y) \to G_0^*(y)$ uniformly in $\mathcal{Y}$. 

Now, we work on $\widehat \btheta$ and $\widehat G_0$. 
By Helly's selection theorem and the Bolzano--Weierstrass theorem, every subsequence of $\{n\}$ has a further subsequence along which $\widehat G_0(y)$ converges pointwise to a limit $G_0^\dagger(y)$, and $\widehat \btheta$ also converges to a limit $\btheta^\dagger$ due to the compactness of $\mathcal{B}$. 
The limit $G_0^\dagger(y)$ is still a CDF because the sequence of $\widehat G_0(y)$ is tight, which is due to the compactness assumption on the response space $\mathcal{Y}$.

Recall that $\ell(\btheta,G_0)$ in \eqref{log-EL} is the log-likelihood function based on the observed data. 
Since $(\widehat \btheta, \widehat G_0)$ maximizes the log-likelihood, we have 
\begin{align}
n^{-1}\ell(\btheta^*, \widetilde G_0) - n^{-1}\ell(\widehat \btheta, \widehat G_0) \leq 0.
\label{proof_KL1}
\end{align}
Following the same reasoning as in \citet{murphy1997maximum} (Appendix A.1, page 974), taking limit on the left-hand side of \eqref{proof_KL1} gives the Kullback--Leibler divergence of the DRM induced by $(\btheta^*, G_0^*)$ from the DRM induced by $(\btheta^\dagger, G_0^\dagger)$.
Given that the Kullback–Leibler divergence is nonnegative, the divergence in this case must be 0, which implies that the two preceding DRMs are identical.
Combined with the identifiability of the DRM in Theorem~\ref{DRM_identifiable} and the assumption \eqref{DRM_identifiable2} on $\bbeta(x;\btheta_k)$, we have $(G_0^\dagger,\btheta^\dagger)=(G_0^*,\btheta^*)$. 
Therefore, we can conclude that $(\widehat \btheta, \widehat G_0)$ uniformly converges to the truth.

\end{proof}

\section{Additional numerical results and details of the experiments in the main text}
\label{supp_details}

\subsection{Data generated from Gaussian distribution}

The true distributions for the potential outcomes are given by: 
\begin{align*}
Y(0) \sim N(2, 2), 
\quad
Y(1) | X_1 \sim \frac{1}{2} N(2+3X_1-0.5X_1^2, 2) + \frac{1}{2} N(2+5X_1-0.5X_1^2, 2),
\end{align*}
with $X_1 \sim N(1,1)$.
We report the numerical performance of the ATE and QTET estimators considered in the main text in Table~\ref{S1_ATE_QTET}.
We also depict the distribution shift from the distribution of $Y(0)$ to that of $Y(1)$ from the CDF plots in Figure~\ref{S1_CDF}.

\begin{table}[h!]
\centering
\caption{Gaussian data. Absolute biases, standard errors (SEs), and root mean square errors (RMSEs) (all after $\times 10$) of various estimators of ATE; and biases and SEs (both after $\times 10$) of various estimators of QTET. 
DRM methods use $\bq(y)=(y, y^2)^\top$ with $\bbeta^\top_{\mathrm{cor}}(x,a)=(x_1, x_1^2, x_2) \btheta_a$, $\bbeta^\top_{\mathrm{mis1}}(x,a)=(x_1, x_2) \btheta_a$, or $\bbeta^\top_{\mathrm{mis2}}(x,a)=x_1 \btheta_a$. 
Other methods either include the full or some partial covariates.}
\label{S1_ATE_QTET}
\resizebox{\textwidth/2}{!}{
\begin{tabular}{lrrr}
\hline
Methods & Abs bias  & SE & RMSE   \\ \hline
{\bf DRM(full)} &1.266&1.519&1.614  \\ \hline
DRM(mis1)             &1.228&1.343&1.527  \\ \hline
DRM(mis2)             &9.946&1.306&10.032  \\ \hline
G-formula(full) &1.028&1.281&1.283 \\ \hline
G-formula(mis1) &1.067&1.350&1.351 \\ \hline
G-formula(mis2) &9.940&1.332&10.029 \\ \hline
IPW(full) &14.027&12.725&18.778 \\ \hline
IPW(mis1) &12.663&5.167&13.676 \\ \hline
IPW(mis2) &9.931&1.333&10.020 \\ \hline
AIPW(full) &1.261&1.935&1.937 \\ \hline
AIPW(mis1) &2.724&2.235&3.423 \\ \hline
AIPW(mis2) &9.931&1.333&10.020 \\ \hline
\end{tabular}
}
\resizebox{\textwidth}{!}{
\begin{tabular}{lrrrrrrrrrr}
\hline
Methods & \multicolumn{2}{c}{10\%} & \multicolumn{2}{c}{30\%} & \multicolumn{2}{c}{50\%} & \multicolumn{2}{c}{70\%} & \multicolumn{2}{c}{90\%}  \\ \hline
& Bias & SE & Bias & SE & Bias & SE & Bias & SE & Bias & SE \\ \hline
{\bf DRM(full)} & -0.31&3.69&-0.10&2.59&-0.20&2.36&-0.13&2.55&-0.43&3.71 \\ \hline
DRM(mis1) & 0.20&3.71&0.25&2.58&-0.25&2.34&-0.50&2.55&-0.69&3.68 \\ \hline
DRM(mis2) & -0.08&3.71&0.25&2.51&-0.05&2.16&-0.20&2.03&-0.31&2.33 \\ \hline
IPW(full) & 9.21&5.43&9.77&3.80&9.86&3.60&9.99&3.47&9.95&3.36 \\ \hline
IPW(mis1) & 9.30&6.09&9.81&3.80&9.74&3.08&9.82&2.76&9.78&2.91 \\ \hline
IPW(mis2) & -0.48&4.05&-0.25&2.65&-0.27&2.19&-0.16&2.09&-0.23&2.51 \\ \hline
CF(QR; full) & 3.58&3.29&1.23&2.35&0.63&2.11&0.34&2.18&0.08&2.68 \\ \hline
CF(QR; mis1) & 4.75&3.13&-2.42&2.32&-3.68&2.10&-2.30&2.21&3.85&2.87 \\ \hline
CF(QR; mis2) & 4.67&3.04&-2.49&2.15&-3.70&1.82&-2.32&1.74&3.92&2.03 \\ \hline
CF(logit; full) & 1.15&3.79&0.18&2.74&0.00&2.49&0.46&2.82&-1.21&5.23 \\ \hline
CF(logit; mis1) & 1.13&3.77&0.19&2.75&0.00&2.49&0.46&2.80&-1.23&5.31 \\ \hline
CF(logit; mis2) & 1.01&3.73&0.18&2.64&0.04&2.14&0.58&2.01&-0.23&2.38 \\ \hline
\end{tabular}
}
\end{table}

\begin{figure}[h!]
\centering 
\includegraphics[width = \textwidth*8/9, keepaspectratio]{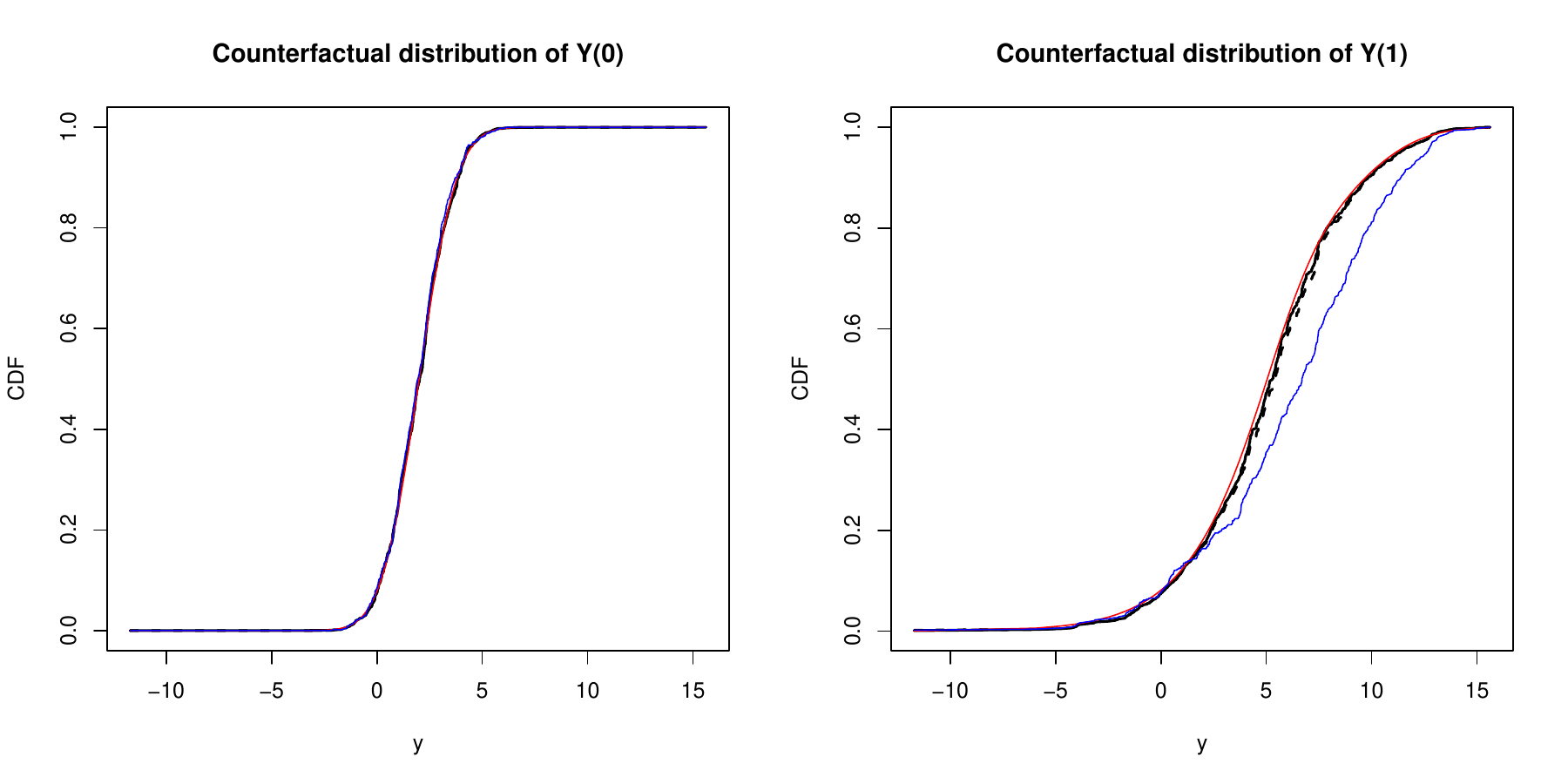}
\caption{CDF plots for $Y(a)$ based on one simulation repetition of Gaussian data. 
All use DRM with $\bq(y)=(y, y^2)^\top$.
Solid black: $\bbeta_{\mathrm{cor}}(x,a)$; 
dashed black: $\bbeta_{\mathrm{mis1}}(x,a)$; 
solid red: the truth; and 
solid blue: the empirical CDF.}
\label{S1_CDF}
\end{figure}

\subsection{Data generated from gamma distribution}

The true distributions for the potential outcomes are given by: 
\begin{align*}
Y(a) | X_1 & \sim \frac{1}{2} \mathrm{Gamma}((2+a)(X_1+1), 0.5(a+1)) + \frac{1}{2} \mathrm{Gamma}((3+a)(X_1+1), 0.5(a+1)), 
\end{align*}
for $a = 0,1$, with $X_1 \sim \mathrm{Gamma}(1,0.5)$.
The numerical results for the performance of the ATE and QTET estimators are provided in Table~\ref{S2_ATE_QTET}. 
The CATE and CDF plots for the DRM method are given in Figure~\ref{S2_CATE}.

\begin{table}[h!]
\centering
\caption{Gamma data. Absolute biases, SEs, and RMSEs (all after $\times 10$) of various estimators of ATE; and biases and SEs (both after $\times 10$) of various estimators of QTET. 
DRM methods use $\bq_{\mathrm{cor}}(y)=(y, \log y)^\top$ or $\bq_{\mathrm{mis}}(y)=(y, y^2)^\top$ with $\bbeta^\top_{\mathrm{cor}}(x,a)=(x_1, x_2) \btheta_a$ or $\bbeta^\top_{\mathrm{mis}}(x,a)=x_1 \btheta_a$. 
Other methods either include the full or some partial covariates.}
\label{S2_ATE_QTET}
\resizebox{\textwidth/2}{!}{
\begin{tabular}{lrrr}
\hline
Methods & Abs bias  & SE & RMSE   \\ \hline
{\bf DRM(full)} &0.861&1.007&1.078  \\ \hline
DRM(mis1) &0.833&1.013&1.041  \\ \hline
DRM(mis2) &11.219&1.253&11.289  \\ \hline
G-formula(full) &0.801&0.996&0.996  \\ \hline
G-formula(mis) &11.225&1.254&11.295  \\ \hline
IPW(full) &18.376&59.812&62.570  \\ \hline
IPW(mis) &11.226&1.252&11.296  \\ \hline
AIPW(full) &1.754&7.624&7.636  \\ \hline
AIPW(mis) &11.225&1.254&11.295  \\ \hline
\end{tabular}
}
\resizebox{\textwidth}{!}{
\begin{tabular}{lrrrrrrrrrr}
\hline
Methods & \multicolumn{2}{c}{10\%} & \multicolumn{2}{c}{30\%} & \multicolumn{2}{c}{50\%} & \multicolumn{2}{c}{70\%} & \multicolumn{2}{c}{90\%}  \\ \hline
& Bias & SE & Bias & SE & Bias & SE & Bias & SE & Bias & SE \\ \hline
{\bf DRM(full)} & 0.12&1.07&-0.31&1.17&-0.38&1.45&-0.03&1.94&1.25&3.50 \\ \hline
DRM(mis1) & 1.30&1.20&0.11&1.19&-1.12&1.49&-1.14&2.05&1.84&3.85 \\ \hline
DRM(mis2) & 4.24&1.19&5.59&1.31&6.73&1.58&8.32&2.07&11.76&3.59 \\ \hline
IPW(full) & -1.56&6.70&-5.57&9.41&-9.84&11.52&-13.24&12.47&-11.53&12.42 \\ \hline
IPW(mis) & 3.94&1.29&5.69&1.36&7.01&1.65&8.52&2.16&11.05&3.67 \\ \hline
CF(QR; full) & 0.95&0.86&0.57&0.94&0.45&1.10&0.35&1.43&0.21&2.51 \\ \hline
CF(QR; mis) & 5.22&1.14&6.45&1.20&7.65&1.40&9.07&1.78&11.60&2.98 \\ \hline
CF(logit; full) & 0.40&1.18&0.17&1.25&0.11&1.50&0.46&2.00&-1.36&4.23 \\ \hline
CF(logit; mis) & 4.40&1.28&5.91&1.38&7.22&1.60&9.25&2.09&11.21&3.58 \\ \hline
\end{tabular}
}
\end{table}

\begin{figure}[h!]
\centering 
\includegraphics[width = \textwidth*4/9, keepaspectratio]{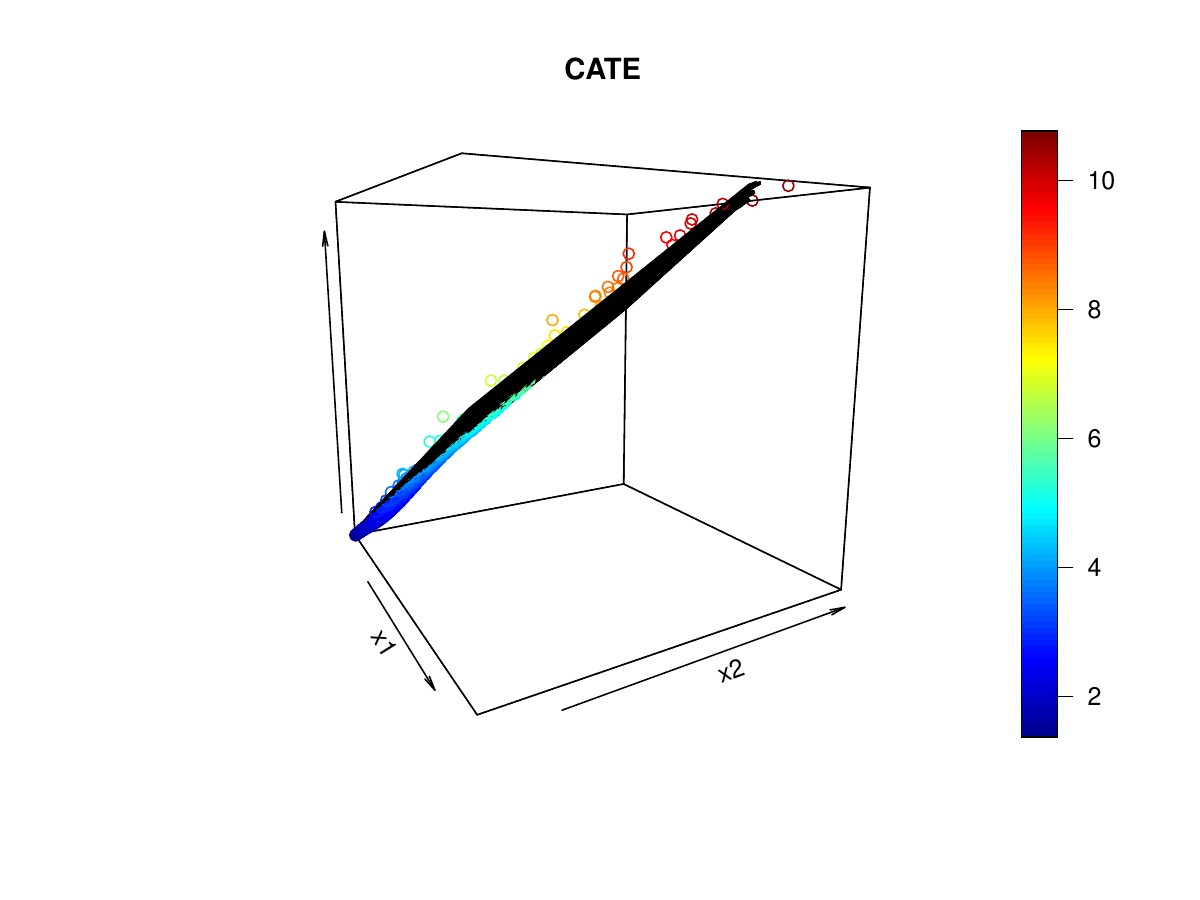}
\includegraphics[width = \textwidth*4/9, keepaspectratio]{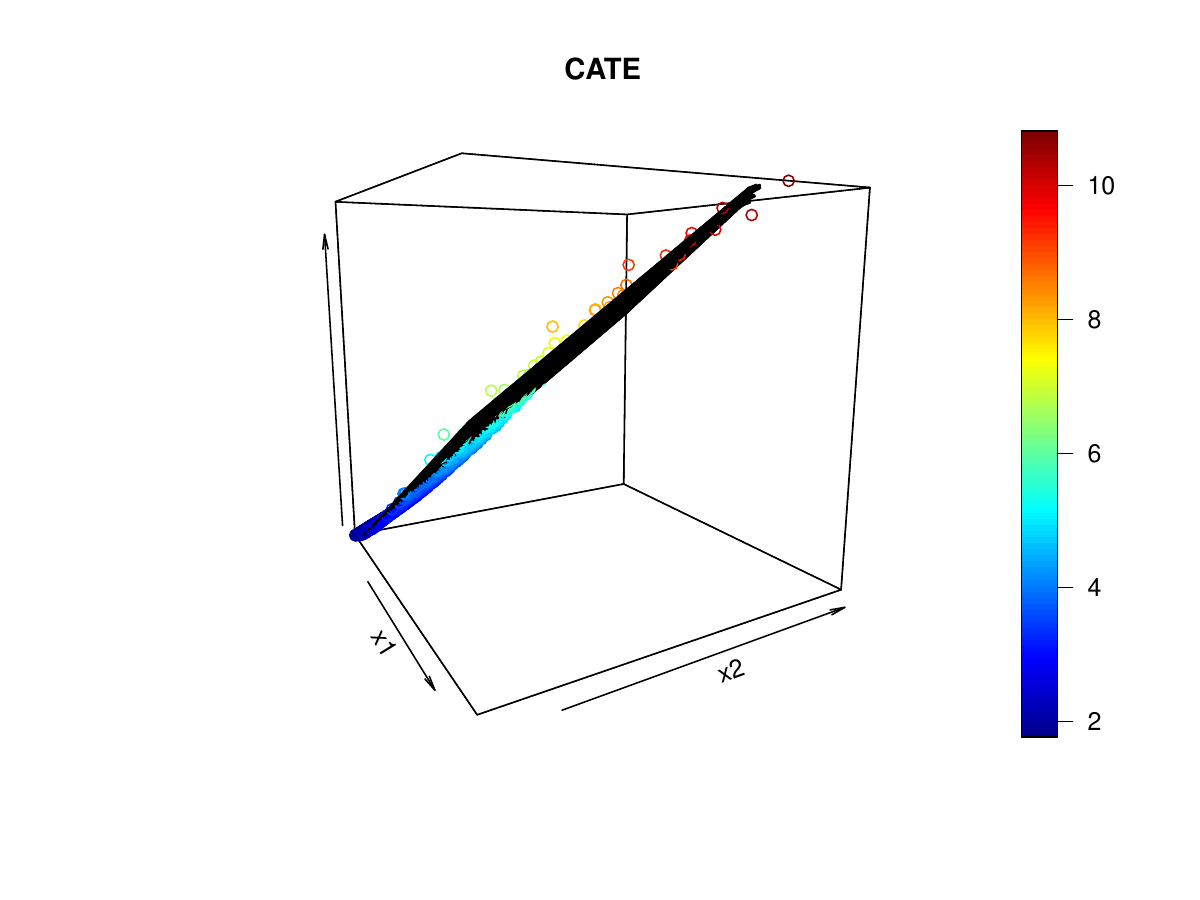} 
\includegraphics[width = \textwidth*8/9, keepaspectratio]{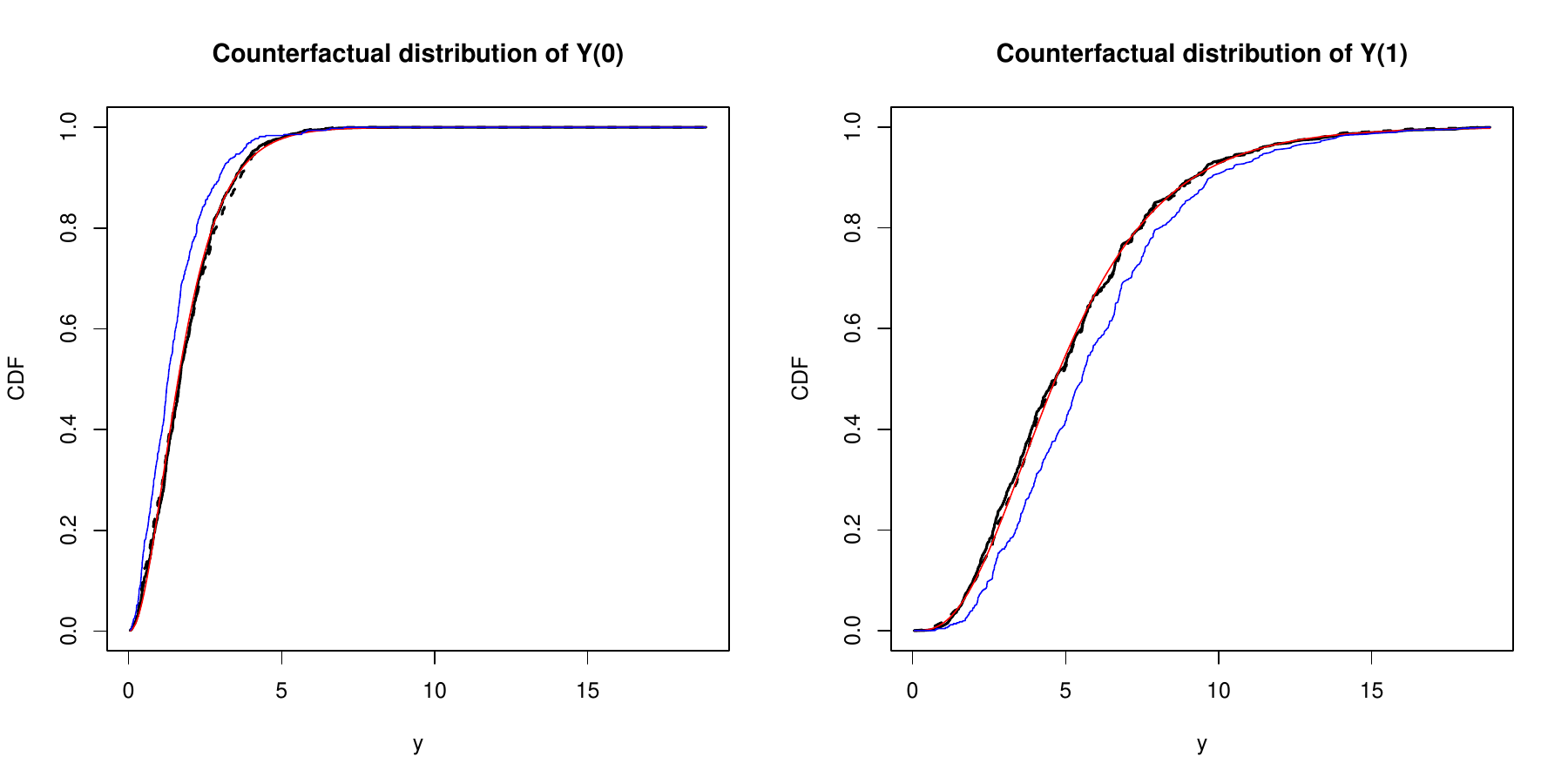}
\caption{CATE and CDF plots based on one simulation repetition of gamma data. 
All use DRM with $\bbeta^\top_{\mathrm{cor}}(x,a)=(x_1, x_2) \btheta_a$.
Top left: with $\bq_{\mathrm{cor}}(y)=(y, \log y)^\top$; 
top right: with $\bq_{\mathrm{mis}}(y)=(y, y^2)^\top$; the black surface depicts the truth.
Bottom: CDF for $Y(a)$, with solid black using $\bq_{\mathrm{cor}}(y)$; dashed black using $\bq_{\mathrm{mis}}(y)$; 
solid red being the truth; and solid blue being the empirical CDF.}
\label{S2_CATE}
\end{figure}

\subsubsection{Real-data analysis: STAR project}

The estimated CATEs for regular versus small classes, grouped by students' demographics, are depicted in Figure~\ref{real_CATE_supp}.
The figure also includes the CDF plot for the estimated counterfactual distributions.

\begin{figure}[h!]
\centering 
\includegraphics[width = \textwidth*4/5, keepaspectratio]{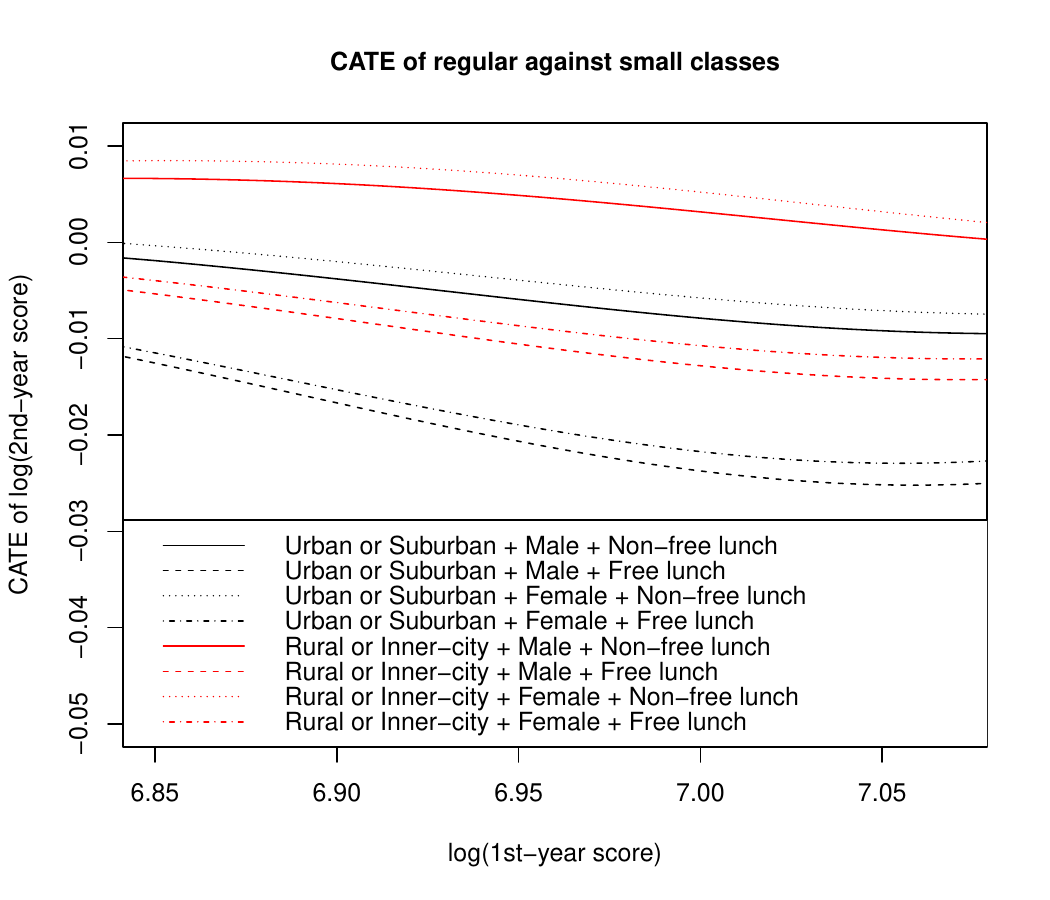}
\includegraphics[width = \textwidth*4/5, keepaspectratio]{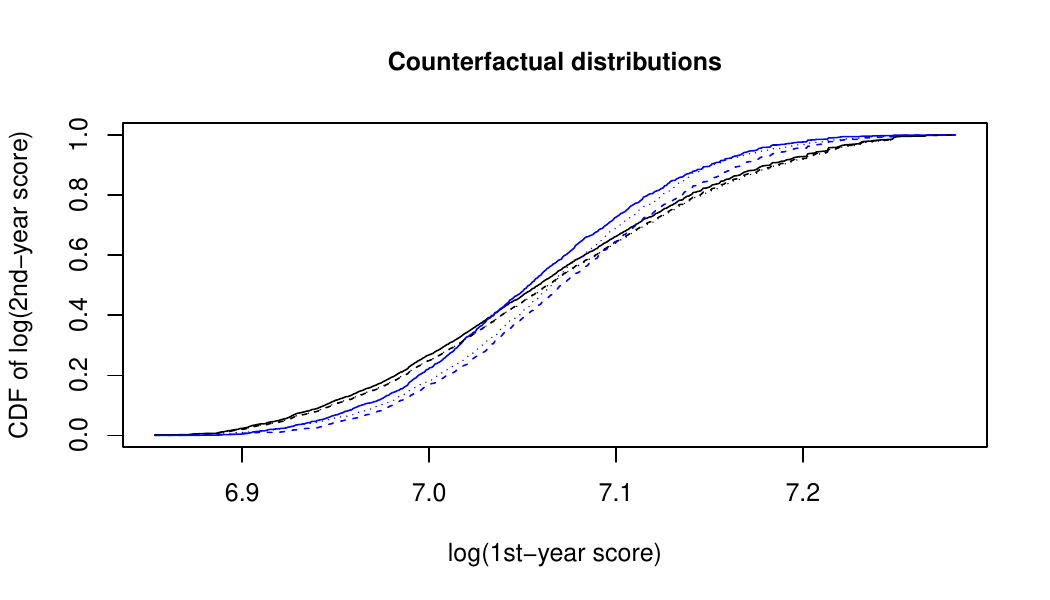}
\caption{Plot of the estimated CATEs for the eight groups of students in STAR real data, and plot of the estimated CDF of the counterfactual distributions. 
For the CDF plot, solid: regular class; dashed: small class; dotted: regular-with-aide class. 
Black: DRM CDF estimator; blue: empirical CDF.}
\label{real_CATE_supp}
\end{figure}

\clearpage

\section{Additional experiments}
\label{supp_add_exp}

\subsection{Data generated from Poisson distribution}

To mimic the scenario of count data, we generate data from the Poisson distribution:
\begin{gather*} 
X_1 \sim \mathrm{Unif} (-1,1), 
\quad
X_2 \sim N(0,1), 
\quad
A \sim \Bernoulli (\pi(x)), 
\, 
\text { with } \pi(x)=\mathrm{ilogit}(0.5-0.5x_1-2x_1^2-0.5x_2),
\\ 
Y |X=x, A=a \sim \mathrm{Pois} (\lambda(x,a)), 
\, 
\text { with } \lambda(x,a)=\exp\{5-0.1(a+1)x1-ax_1^2-0.1(a+1)x_2)\}.
\end{gather*} 
The true ATE is $-35.753$, and the true QTET at levels $p=0.1, 0.3, 0.5, 0.7, 0.9$ are $(-50,-36,-26,-16,-1)$.
For all the methods, we consider both using the full set of covariates $(x_1, x_1^2, x_2)$ (``full'') and using only covariate $x_2$ (``mis'').
In the DRM, we use the best-performing basis function we find: $\bq(y)=(\sqrt{y},y)^\top$, which leads to a correct specification, and we specify $\bbeta(x,a)$ as a linear function of the covariates included in the model. 
The results are provided in Figure~\ref{S4_boxplot} and Table~\ref{S4_ATE_QTET}. 
We observe that when the full set of covariates is used, the DRM-based ATE estimator outperforms the estimators of all other methods. The DRM-based QTET estimators also perform well, typically exhibiting small biases. However, when only $x_2$ is used, all methods tend to produce much larger biases, indicating that $x_1$ may contain crucial information about the causal effects in this scenario.

\subsection{Data generated from exponential distribution}

Finally, we also generate data from the exponential distribution $\mathrm{Exp} (\lambda)$ with $\lambda$ being the rate parameter:
\begin{gather*} 
X_1 \sim \mathrm{Unif} (-1,1), 
\quad
X_2 \sim \mathrm{Exp} (1), 
\quad
A \sim \Bernoulli (\pi(x)), 
\, 
\text { with } \pi(x)=\mathrm{ilogit}(1-x_1+0.5x_2-x_1x_2),
\\ 
Y |X=x, A=a \sim \mathrm{Exp} (\lambda(x,a)), 
\, 
\text { with } \lambda(x,a)=0.1\{1+a(x_1+1)+0.5(a+1)x_2+(x_1+1)x_2\}.
\end{gather*} 
The true ATE is $-2.063$, and the true QTET at levels $0.1, 0.3, 0.5, 0.7, 0.9$ are $(-0.167, -0.600, -1.245, -2.335, -4.927)$.
We fit the DRM to the data with the best-performing basis function we find: $\bq(y)=\sqrt{y}$. 
For $\bbeta(x,a)$ specification, we consider a full model with $\bbeta^\top_{\mathrm{cor}}(x,a)=(x_1, x_2, x_1x_2) \btheta_a$ (``full'') and a partial model with the interaction term $x_1x_2$ dropped: $\bbeta^\top_{\mathrm{mis}}(x,a)=(x_1, x_2) \btheta_a$ (``mis''). 
For other competing methods, we also consider the full model that uses $(x_1, x_2, x_1x_2)$ (``full'') and a model that uses only $(x_1, x_2, x_1x_2)$ (``mis'').
The results are provided in Figure~\ref{S3_boxplot} and Table~\ref{S3_ATE_QTET}.
For both ATE and QTET estimation, the DRM estimators demonstrate satisfactory performance whether the model is correctly specified or misspecified. Interestingly, all methods tend to produce smaller or comparable RMSEs when $x_1x_2$ is excluded from the model. This suggests that, in this case, the interaction term does not significantly impact the causal effects.

\clearpage
\begin{figure}[H]
\centering 
\includegraphics[width=\textwidth,keepaspectratio]{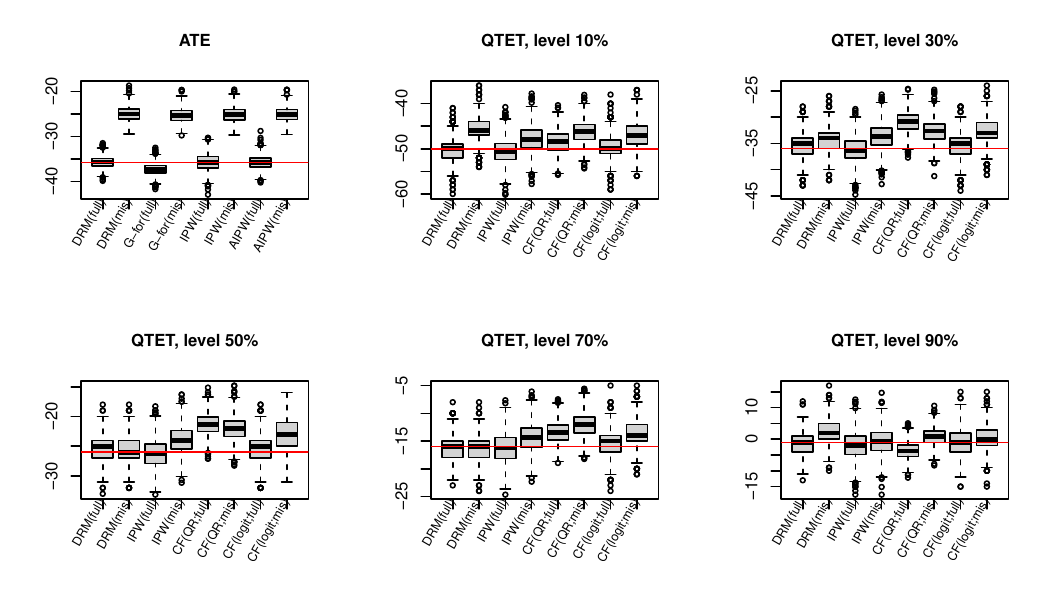}
\caption{Boxplots of the ATE and QTET estimators using various methods based on 1000 simulation repetitions of Poisson data. 
The red lines depict the truth.}
\label{S4_boxplot}
\end{figure}

\begin{table}[H]
\centering
\caption{Poisson data. Absolute biases, SEs, and RMSEs of various estimators of ATE; and biases and SEs of various estimators of QTET. 
Other methods either include the full or some partial covariates.}
\label{S4_ATE_QTET}
\resizebox{\textwidth/2}{!}{
\begin{tabular}{lrrr}
\hline
Methods & Abs bias  & SE & RMSE   \\ \hline
{\bf DRM(full)} &1.003&1.288&1.288  \\ \hline
DRM(mis) &10.772&1.608&10.892  \\ \hline
G-formula(full) &1.737&1.360&2.075  \\ \hline
G-formula(mis) &10.453&1.595&10.574  \\ \hline
IPW(full) &1.501&1.889&1.889  \\ \hline
IPW(mis) &10.608&1.725&10.747  \\ \hline
AIPW(full) &1.177&1.506&1.506  \\ \hline
AIPW(mis) &10.604&1.608&10.726 \\ \hline
\end{tabular}
}
\resizebox{\textwidth}{!}{
\begin{tabular}{lrrrrrrrrrr}
\hline
Methods & \multicolumn{2}{c}{10\%} & \multicolumn{2}{c}{30\%} & \multicolumn{2}{c}{50\%} & \multicolumn{2}{c}{70\%} & \multicolumn{2}{c}{90\%}  \\ \hline
& Bias & SE & Bias & SE & Bias & SE & Bias & SE & Bias & SE \\ \hline
{\bf DRM(full)} & -0.36&2.85&0.81&2.48&0.59&2.22&-0.10&2.33&-0.31&3.34 \\ \hline
DRM(mis) & 4.37&2.83&1.73&2.44&0.12&2.18&-0.41&2.38&3.35&3.75 \\ \hline
IPW(full) & -0.57&2.79&-0.22&2.52&-0.20&2.48&-0.30&2.75&-0.86&4.43 \\ \hline
IPW(mis) & 2.17&2.94&2.33&2.48&2.09&2.37&1.62&2.59&0.37&4.12 \\ \hline
CF(QR; full) & 1.53&2.55&5.07&2.02&4.67&1.87&2.53&2.02&-2.62&2.74 \\ \hline
CF(QR; mis) & 3.74&2.45&3.31&2.16&3.97&2.07&3.99&2.14&1.82&2.72 \\ \hline
CF(logit; full) & 0.48&2.90&0.67&2.55&0.63&2.44&0.48&2.61&0.10&3.94 \\ \hline
CF(logit; mis) & 3.18&2.99&3.22&2.55&2.95&2.45&2.44&2.60&1.25&3.90 \\ \hline
\end{tabular}
}
\end{table}

\clearpage
\begin{figure}[H]
\centering 
\includegraphics[width=\textwidth,keepaspectratio]{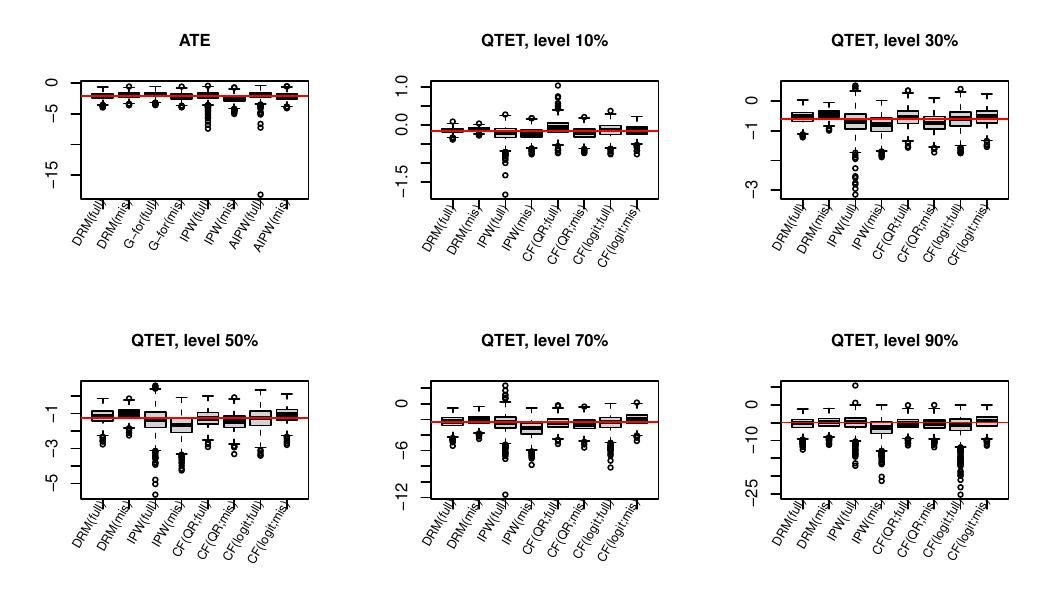}
\caption{Boxplots of the ATE and QTET estimators using various methods based on 1000 simulation repetitions of exponential data. 
The red lines depict the truth.}
\label{S3_boxplot}
\end{figure}

\begin{table}[H]
\centering
\caption{Exponential data. Absolute biases, SEs, and RMSEs (all after $\times 10$) of various estimators of ATE; and biases and SEs (both after $\times 10$) of various estimators of QTET. 
Other methods either include the full or some partial covariates.}
\label{S3_ATE_QTET}
\resizebox{\textwidth/2}{!}{
\begin{tabular}{lrrr}
\hline
Methods & Abs bias  & SE & RMSE   \\ \hline
{\bf DRM(full)} &4.392&5.524&5.535  \\ \hline
DRM(mis) &4.168&5.052&5.141  \\ \hline
G-formula(full) &4.354&5.039&5.412  \\ \hline
G-formula(mis) &4.481&5.549&5.660  \\ \hline
IPW(full) &5.053&7.078&7.079  \\ \hline
IPW(mis) &6.343&6.768&8.104  \\ \hline
AIPW(full) &5.231&8.418&8.474  \\ \hline
AIPW(mis) &4.953&6.123&6.244  \\ \hline
\end{tabular}
}
\resizebox{\textwidth}{!}{
\begin{tabular}{lrrrrrrrrrr}
\hline
Methods & \multicolumn{2}{c}{10\%} & \multicolumn{2}{c}{30\%} & \multicolumn{2}{c}{50\%} & \multicolumn{2}{c}{70\%} & \multicolumn{2}{c}{90\%}  \\ \hline
& Bias & SE & Bias & SE & Bias & SE & Bias & SE & Bias & SE \\ \hline
{\bf DRM(full)} & 0.23&0.72&0.64&2.16&0.81&4.25&0.44&7.54&-3.49&18.06 \\ \hline
DRM(mis) & 0.60&0.49&1.67&1.53&2.52&3.23&2.88&6.31&-0.28&16.63 \\ \hline
IPW(full) & -0.58&1.97&-0.93&4.27&-1.29&7.29&-1.38&11.72&-1.46&22.46 \\ \hline
IPW(mis) & -0.64&1.46&-2.03&3.33&-4.60&6.36&-8.68&11.07&-16.84&26.41 \\ \hline
CF(QR; full) & 0.98&1.84&0.54&3.04&-0.31&4.73&-1.40&7.72&-2.95&16.56 \\ \hline
CF(QR; mis) & -0.48&1.55&-1.32&3.09&-2.23&5.01&-3.12&8.11&-4.63&17.40 \\ \hline
CF(logit; full) & 0.29&1.79&-0.01&3.60&-0.38&6.17&-1.11&10.61&-10.11&30.04 \\ \hline
CF(logit; mis) & 0.12&1.42&0.64&2.88&1.76&4.65&3.52&7.73&2.32&18.18 \\ \hline
\end{tabular}
}
\end{table}

\end{document}